\documentclass[sn-mathphys-num]{sn-jnl}

\usepackage{graphicx}
\usepackage{multirow}
\usepackage{amsmath,amssymb,amsfonts}
\usepackage{amsthm}
\usepackage{mathrsfs}
\usepackage[title]{appendix}
\usepackage{xcolor}
\usepackage{textcomp}
\usepackage{manyfoot}
\usepackage{booktabs}
\usepackage{algorithm}
\usepackage{algorithmicx}
\usepackage{algpseudocode}
\usepackage{listings}
\usepackage{physics}
\usepackage{bm}
\usepackage[utf8]{inputenc}
\usepackage{ulem}
\usepackage{natbib}
\usepackage{tikz-cd}

\theoremstyle{plain}
\newtheorem{theorem}{Theorem}
\newtheorem{proposition}[theorem]{Proposition}

\theoremstyle{remark}

\theoremstyle{definition}

\begin{document}

\title{Dirac points annihilation and its obstruction characterized by Euler number and quaternionic charges in kagome lattice}

\author*[1,2]{\fnm{M.} \sur{Finck}}\email{matthieu.finck@uca.fr}
\author[1,3]{\fnm{D.} \sur{Solnyshkov}}
\author[2]{\fnm{J.} \sur{Dubois}}
\author[1]{\fnm{G.} \sur{Malpuech}}

\affil[1]{\orgdiv{Institut Pascal}, \orgname{Université Clermont Auvergne, CNRS}, \orgaddress{\city{Clermont-Ferrand}, \postcode{F-63000}, \country{France}}}

\affil[2]{\orgdiv{Laboratoire de Mathématiques Blaise Pascal}, \orgname{Université Clermont Auvergne, CNRS}, \orgaddress{\city{Clermont-Ferrand}, \postcode{F-63000}, \country{France}}}

\affil[3]{\orgname{Institut Universitaire de France (IUF)}, \orgaddress{\city{Paris}, \postcode{F-75231}, \country{France}}}

\maketitle 

\begin{abstract}
~We investigate the topological phenomenon of Dirac point annihilation and its obstruction in three-band, real symmetric Hamiltonians with time-reversal symmetry, and their relation to the Euler number, a well-known topological invariant.
For this purpose, we study the example of the kagome lattice using a simple tight-binding model. By tuning the parameters of the lattice continuously, we illustrate situations where two Dirac points are able to annihilate, and others, where this annihilation is topologically obstructed. For a system with no gaps between the three bands, like in the kagome lattice, the Euler number of two bands is ill-defined on the whole Brillouin zone, which requires the introduction of the so-called ``patch" Euler number on a subregion without additional degeneracies coming from the third band. A non-zero patch Euler number means that the annihilation of the Dirac points is impossible.
We also illustrate another point of view, using homotopy theory, associating the Dirac points with quaternionic charges. We prove that the non-abelian braiding of the Dirac points in k-space conjugates their quaternionic charge and explains the possible obstruction to the annihilation of Dirac points.
Finally, we show that the proposed deformation of the kagome lattice can be achieved in realistic photonic systems.
\end{abstract}


\section{Introduction}\label{sec1}

The geometric phase in quantum mechanics, or Berry phase, was introduced by Michael Berry in 1984~\cite{berry_quantal_1984}. This concept appears in the Adiabatic Theorem~\cite{sakurai_modern_2011} when studying quantum systems responding to a slow, cyclic change in their parameters~\cite{shapere_geometric_1989}. While phases in quantum mechanics were well-known since its beginnings, the geometric phase revealed an overlooked contribution, which depends only on the geometry of the parameter space, independent of time and energy. 

More precisely, the geometric phase factor happens to be the holonomy of a connection on a fiber bundle~\cite{nakahara_geometry_2003}, the Berry connection, where the bundle's base space is the manifold of parameters and the fiber represents the quantum state space. In other words, while slowly moving inside the parameter space, the state of the system evolves according to the parallel transport defined by this connection. This formulation revealed a link between quantum mechanics and gauge theories, in which the Berry connection behaves mathematically like a gauge potential, as well as a link with differential geometry, with the language of fiber bundles.

A widely used concept in topology is the one of topological invariant: roughly speaking, it is a mathematical object (number, polynomial, group,...) which remains unchanged under continuous deformations without tearing or introducing singularities. For fiber bundles, it is well-known that characteristic classes and characteristic numbers are efficient invariants to classify fiber bundles~\cite{milnor_characteristic_nodate}. They detect to which extent a fiber bundle is non-trivial, meaning that it is not possible to find a global choice of phase for the state vectors. A well-known characteristic number, the first Chern number, can be computed by integrating the Berry curvature over the Brillouin zone.

This framework has proven effective in understanding topological phenomena such as the quantum Hall effect~\cite{klitzing_new_1980}, magnetic monopoles, and new phases of matter such as topological insulators~\cite{kane_quantum_2005} which exhibit interesting topological effects, such as chiral edge states that are robust to small perturbations of the material.

Another topological phenomenon which can be studied with fiber bundles, and the topic of this work, is the annihilation of Dirac points. Dirac points are degeneracies between two bands where the dispersion becomes linear, which produces massless quasiparticles described by the Dirac equation~\cite{castro_neto_electronic_2009}. Dirac points usually come in pairs, and when they are brought together, they might annihilate, opening a gap between the two bands. This annihilation signals a topological transition between a
semimetallic phase and a band insulator~\cite{mergingOfDirac}. However, it is also possible that
the two points collide and then separate again without disappearing, in which case we say that the annihilation is obstructed. The study of the annihilation of Dirac points was first carried out in two-band models in~\cite{mergingOfDirac}. Our work on three-band Dirac point annihilation builds upon recent developments by~\cite{wu_non-abelian_2019},~\cite{bouhon_non-abelian_2020} and~\cite{ahn_failure_2019}. 

In this work, we consider the kagome lattice (whose name originates from the traditional woven bamboo pattern~\cite{tilings}), which we describe in Section~\ref{sec:kagome}. The advantage of this model is that the annihilation and its obstruction can be studied by varying experimentally accessible parameters, contrary to the systems studied previously. We describe this lattice with a simple tight-binding nearest-neighbor $3\times 3$ Bloch Hamiltonian, which depends on the 2D wave vector $k$. In Section~\ref{annihilation}, we consider a specific deformation of the kagome lattice where two Dirac points collide and either annihilate or not depending on the path they take, and we start analyzing the phenomenon by introducing a vector bundle~\cite{milnor_characteristic_nodate} of eigenstates, called the eigenspace bundle. 

A key contribution of our work is proving rigorously the effectiveness of the vector bundle approach for analyzing Dirac points annihilation, using cohomology and characteristic classes theory. We argue that these principles are not only applicable to this specific model, but to a broad class of topological systems where the structure of eigenspaces plays an important role. In Section~\ref{Eulernumber}, we analyze its properties using the Euler number. In Section~\ref{sec:quaternions}, we illustrate the alternative point of view of quaternionic charge and non-abelian reciprocal braiding. In Section~\ref{sec6}, we show how the annihilation and its obstruction can be studied in realistic photonic systems.

\section{The model: the kagome lattice}\label{sec:kagome}

In this section, we describe the kagome lattice, chosen as a realistic model where the parameters responsible for the annihilation and its obstruction can be easily tunable experimentally. 
It is a two-dimensional lattice consisting of equilateral triangles and regular hexagons as shown in Fig.~\ref{fig:kagome}. Just like the honeycomb lattice for graphene, the kagome lattice is not a Bravais lattice~\cite{castro_neto_electronic_2009}, but it can still be described as a lattice with multiple (in this case, three) inequivalent atoms $A$, $B$, $C$ per primitive cell.
We choose a rhombus as a primitive cell.

\begin{figure}[h!]
    \centering
    \includegraphics[width=0.9\linewidth]{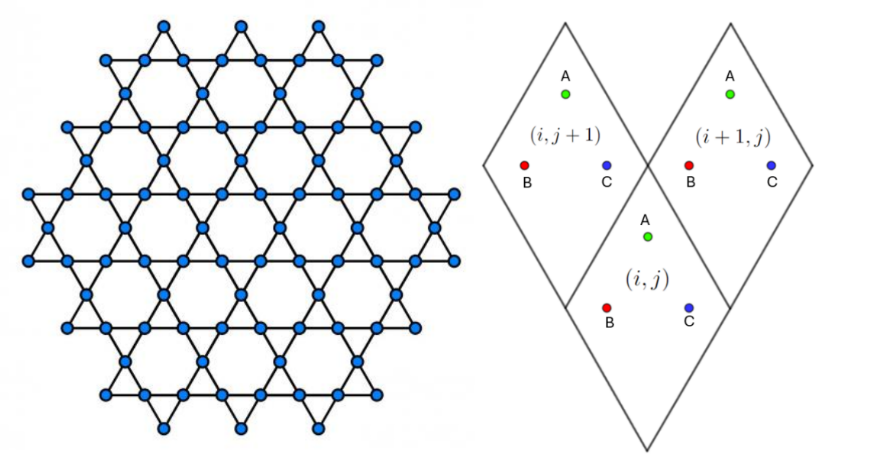}
    \caption{The kagome lattice, and three primitive cells (also called fundamental domains) labeled with two integers corresponding to the two dimensions of the lattice. Each primitive cell contains three sites labeled $A$, $B$ and $C$.}
    \label{fig:kagome}
\end{figure}

In order to model the behavior of the system, we use a simple tight-binding model with nearest-neighbor hopping written in the so-called base II representation~\cite{bena_remarks_2009}, whose advantage for kagome lattice is that the resulting Hamiltonian is purely real~\cite{lim_dirac_2020}.
The details of the construction for the model can be found in Appendix~\ref{annexeKagome}.

We get a $3\times 3$ matrix, the Bloch Hamiltonian~\cite{lim_dirac_2020}:
$$H(k)=-2t\begin{pmatrix} 
0 & \cos(k \cdot \delta_1) & \cos(k \cdot \delta_2) \\ 
\cos(k \cdot \delta_1) & 0 & \cos(k \cdot \delta_3) \\ 
\cos(k \cdot \delta_2) & \cos(k \cdot \delta_3) & 0
\end{pmatrix} $$
where $k$ is a 2D wave vector, $t>0$ is the hopping amplitude and $\delta_1,\delta_2,\delta_3$ are the nearest-neighbor hopping vectors.

$H(k)$ is a real symmetric matrix, hence it is orthogonally diagonalizable with real eigenvalues. It possesses the time-reversal symmetry $H(k)=H(-k)$ (Appendix~\ref{annexeTimereversal}). We order the three bands by increasing energy, $E_1(k) \leq E_2(k) \leq E_3(k)$ for all $k$.
We will denote the corresponding eigenvectors as $u_1(k), u_2(k),u_3(k)$. However, they are not uniquely defined: if $u_i(k)$ is a real normalized eigenvector, then so is $-u_i(k)$. This choice of a $\pm 1$ gauge will play an important role in the geometry of the system.

By plotting the dispersion, we observe three bands exhibiting degeneracies, like in Fig.~\ref{fig:annihilation}.
The degeneracies we are most interested in are the Dirac points: they are defined as band crossing points where the dispersion is linear, which results in cones. Those points have been studied extensively, especially in the context of graphene, because electrons near the Dirac points behave as massless Dirac fermions~\cite{castro_neto_electronic_2009}.

\section{Dirac point annihilation}\label{annihilation}

By allowing the tuning of the on-site energies $E_A,E_B,E_C$ and the tunneling coefficients $t_{AB},t_{AC},t_{BC}$ for different pairs of atoms, the Bloch Hamiltonian $H(k)$ takes the following form:  \begin{equation}
H(k)=-\begin{pmatrix} 
E_A & t_{AB}\cos(k \cdot \delta_1) & t_{AC}\cos(k \cdot \delta_2) \\ 
t_{AB}\cos(k \cdot \delta_1) & E_B & t_{BC}\cos(k \cdot \delta_3) \\ 
t_{AC}\cos(k \cdot \delta_2) & t_{BC}\cos(k \cdot \delta_3) & E_C
\end{pmatrix}
\label{Hamiltonien}
\end{equation} Note that the matrix remains real symmetric and conserves time-reversal symmetry $H(k)=H(-k)$. In solid state systems, the parameters can be tuned to a certain extent by applying different types of strain~\cite{Lima2023} or by considering different materials from the same family of kagome lattices~\cite{Jung2022}. Independent tuning of all parameters is possible in photonic lattices~\cite{li2020higher}.

When tuning some of these parameters, it is possible to bring together the Dirac points, and observe their collision, as in Fig.~\ref{fig:annihilation}. If a gap is opened, and the points disappear, it is called an annihilation. On the contrary, if no gap is opened and the points remain present, we say that the annihilation was obstructed.

\begin{figure}[h!]
    \centering
    \includegraphics[width=1\linewidth]{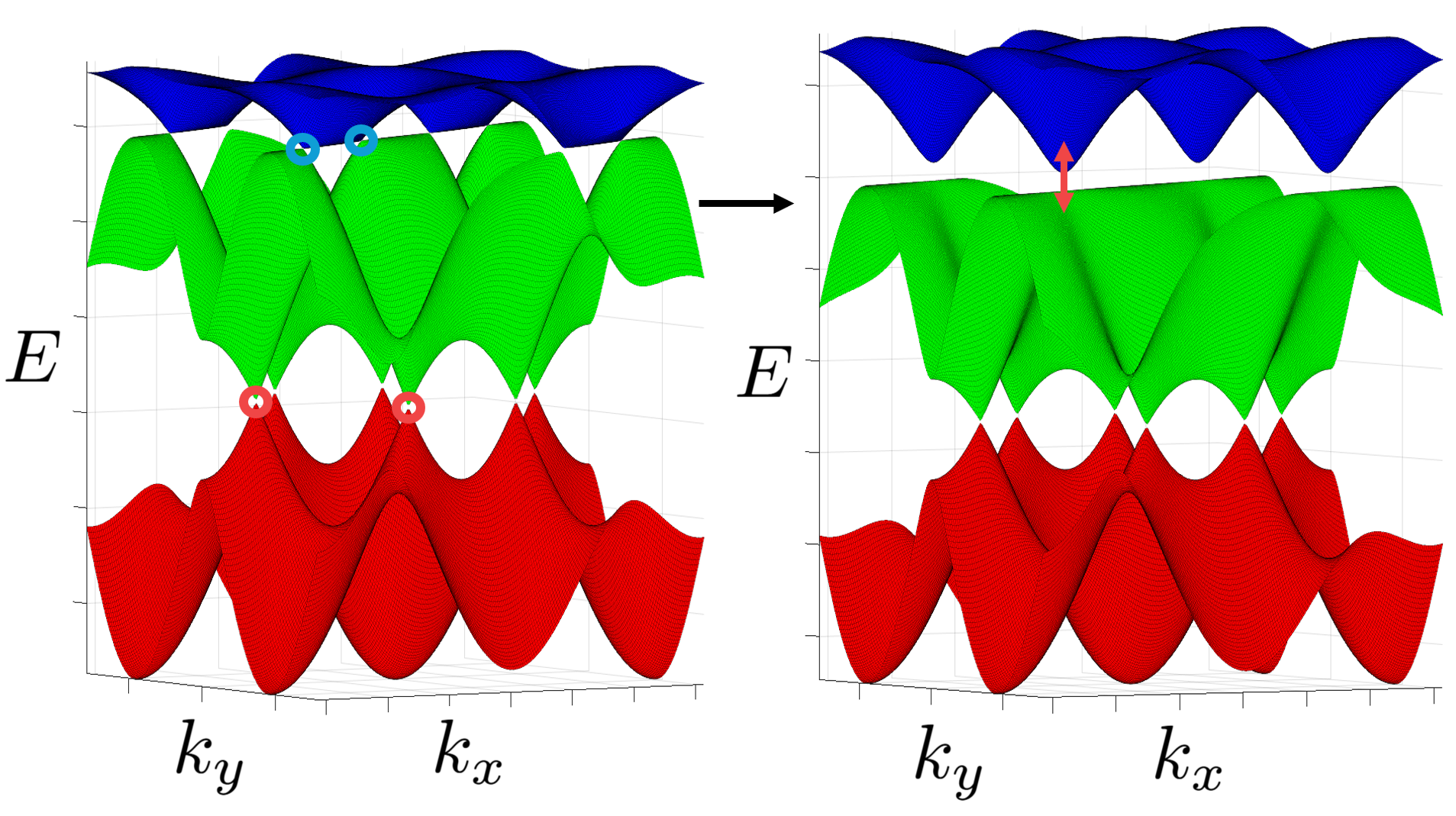}
    \caption{The unobstructed annihilation of principal Dirac points (two of them shown in blue), creating a gap between bands 2 and 3. Two adjacent Dirac points between the other bands are shown in red.}
    \label{fig:annihilation}
\end{figure}

In three-band models, it is essential to distinguish the Dirac points coming from the two different pairs of bands. The Dirac points coming from the two bands of interest will be called ``principal Dirac points" whereas the others will be called ``adjacent Dirac points". Therefore, we will study only the annihilation of principal Dirac points, although it can also happen for adjacent Dirac points. Without loss of generality, in this article, we will only be interested in the case where the principal Dirac points are taken to be the ones between the higher energy bands, 2 and 3; and the adjacent Dirac points between the two lower bands, 1 and 2.

The following 6-panel Fig.~\ref{fig:6panel} is a concrete example which illustrates that two principal Dirac points cannot always annihilate each other, and that this obstruction may or may not occur depending on the presence of adjacent Dirac points. We use a color plot graph of the difference of energies between bands 2 and 3 to visualize the phenomenon more clearly. We start in panel $(a)$ with two principal Dirac points (in dark blue). We first modify one of the tunneling coefficients, $t_{BC}$. The Dirac points merge in $(b)$, and fail to annihilate, as they bounce in the orthogonal direction in $(c)$.  In $(d)$ and $(e)$, we modify one of the on-site energies, $E_A$, which makes adjacent Dirac points appear within the Brillouin zone (BZ). The adjacent Dirac points are not visible in Fig.~\ref{fig:6panel}, but we show and discuss them in detail below.
When the principal Dirac points collide again in $(f)$, they annihilate, and a gap is opened.

\begin{figure}[h!]
    \centering
    \includegraphics[width=1\linewidth]{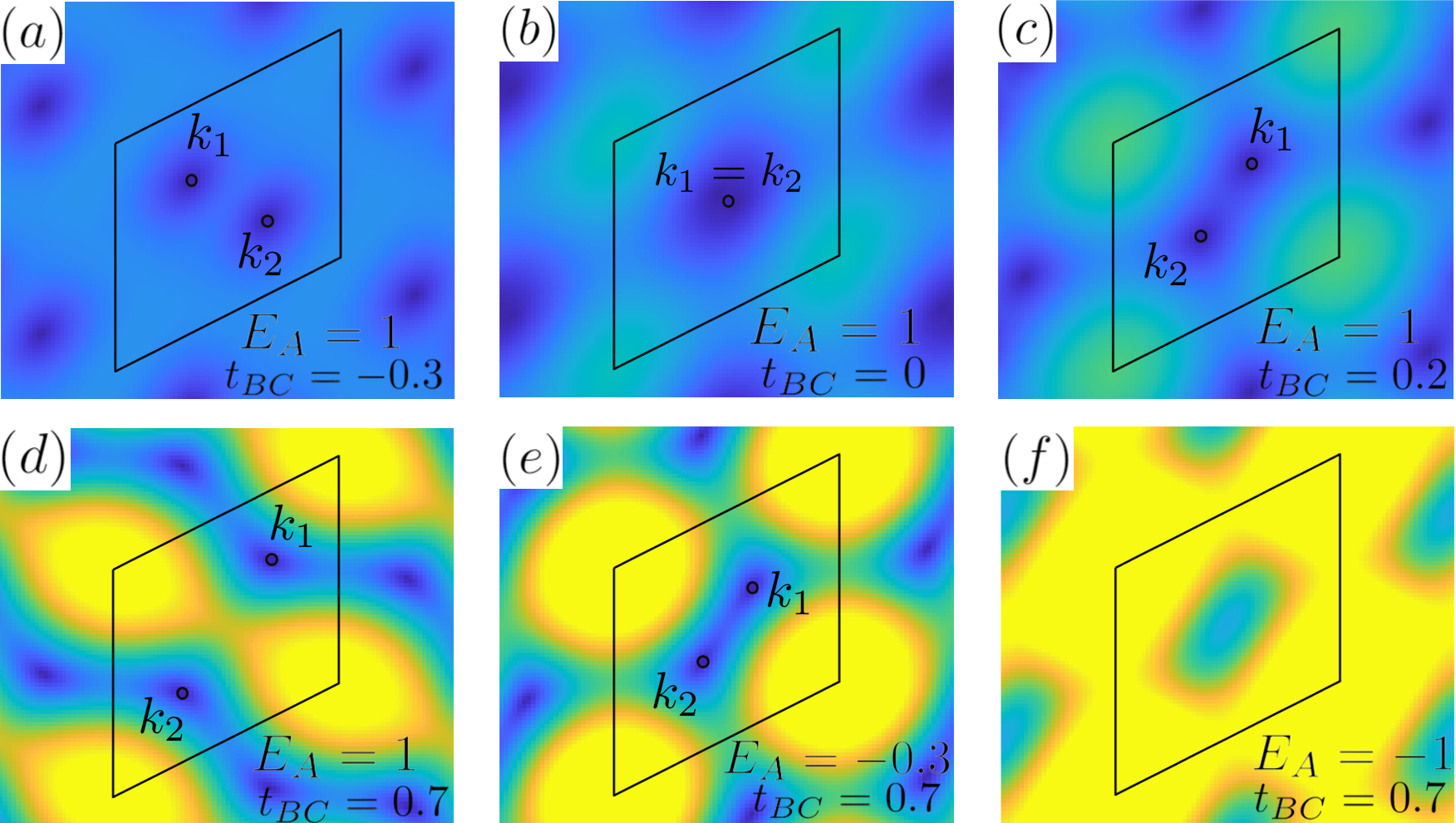}
    \caption{Energy difference between bands 2 and 3 at six different values of the tuning parameters. The rhombus is a choice of a reciprocal primitive cell. The Dirac points $k_1$ and $k_2$ are the points where the energy difference is zero (in dark blue). In panel (a), the initial configuration for the tuning parameters is 
 $E_A=1$, $t_{AB}=1$, $t_{AC}=-1$, $E_{B}=0$,  $t_{BC}=-0.3$, $E_C=0$. The parameters are modified as shown in the bottom right corner of each panel: $t_{BC} =0$,  $t_{BC}=0.2$,  $t_{BC}=0.7$, $E_A=-0.3$, $E_A=-1$.}
    \label{fig:6panel}
\end{figure}

In order to study the topology and geometry of this system, the relevant notion is a rank-2 eigenspace bundle~\cite{sergeevVectorBundle2023}. Vector bundles are reviewed and their construction is explained in Appendix~\ref{annexeEuler}.

For each Hamiltonian $H(k)$, it is possible to create different vector bundles, called the eigenspace bundles, by attaching, to each $k$ in the Brillouin zone (BZ), any number of eigenspaces. In our case, because the annihilation of Dirac points involves two bands, it is natural to study the vector bundle $E$ corresponding to the direct sum of the two eigenspaces of interest, denoted as $F$, as shown in figure~\ref{fig:vectorBundle}. The projection $\pi: E \rightarrow BZ$ associates to each eigenstate the corresponding $k$ in the Brillouin zone. We summarize this data with the following sequence: $$F \xrightarrow{} E\xrightarrow{\pi}BZ$$

\begin{figure}[h!]
    \centering
    \includegraphics[width=0.8\linewidth]{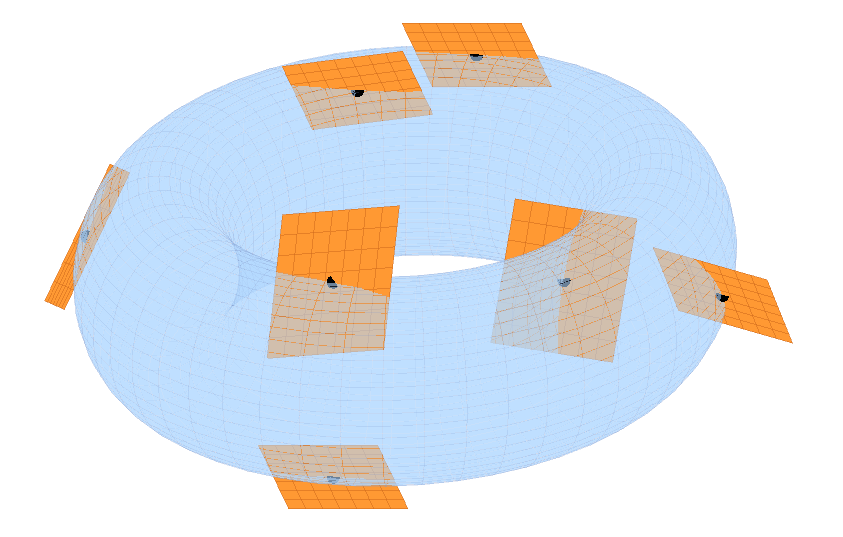}
    \caption{The rank-2 eigenspace bundle $E$ associated with the two principal bands, in the initial configuration of the Bloch Hamiltonian for the kagome lattice. The light blue torus is the Brillouin zone $BZ$, acting as the base space of this bundle. A few fibers are partially represented in orange, each attached to a point (in black) on the torus.}
    \label{fig:vectorBundle}
\end{figure}
By modifying the parameter values $E_A,E_B,E_C,t_{AB},t_{AC},t_{BC}$ of the Bloch Hamiltonian, we get different vector bundles.
A key point is that the two vector bundles before and after a continuous modification of the parameters are isomorphic, as long as the third band stays isolated~\cite{husemoller1994fibre}. This claim is proven in Appendix~\ref{annexeEuler}.
 
We now explain why the phenomenon of Dirac points annihilation, in the context of a system possessing time-reversal symmetry, can be understood with this rank-2 eigenspace bundle. To the best of our knowledge, our approach with bundles has not been fully explored in other articles, and represents one of the main contributions of this work. We also believe that the arguments used are very general, and can be used in many different contexts.

We prove the following proposition, valid if the first band is isolated from the last two :
\begin{proposition}
    If $E_1(k)<E_2(k)\leq E_3(k)$ for all $k\in BZ$, and if the rank-2 bundle associated with bands 2 and 3 is non-trivial, then the principal Dirac points cannot annihilate.
    \label{propositionBundle}
\end{proposition} 

\begin{proof}
    We prove the contrapositive statement: if the principal Dirac points can annihilate, then the rank-2 bundle is trivial.
    
    Suppose that the Dirac points can be annihilated. This means that  the two bands are gapped after modifying the parameters. This gap allows us to consider the rank-1 eigenspace bundles associated with the second and third eigenvalues, respectively. In systems with time-reversal symmetry, such as the kagome lattice, those rank-1 bundles are guaranteed to be trivial. (This property, proven in detail in Appendix~\ref{annexeTimereversal}, does not hold for general complex vector bundles, nor for real vector bundles in the absence of symmetries). Subsequently, the rank-2 bundle can be decomposed as a Whitney sum (direct sum) of two trivial bundles, making it trivial. This shows the contrapositive statement, and in turn the proposition.
\end{proof}


Therefore, a natural question to ask is when the eigenspace bundle is trivial. In order to detect the triviality of a vector bundle, as we do traditionally in topology, we use topological invariants. The relevant notions of topological invariants for vector bundles are characteristic classes and characteristic numbers~\cite{milnor_characteristic_nodate}, which measure the ``twisting" in vector bundles.

\section{The Euler number}\label{Eulernumber}

The Euler class of a real, oriented fiber bundle is a characteristic class~\cite{tu_differential_2017}, which appears naturally when studying the obstruction to define a global section on the whole Brillouin zone, that is to say, an eigenstate $u(k)$ defined for every $k$~\cite{HatcherBundle}.

For our purposes, integrating the Euler class on the Brillouin zone yields the Euler number, or Euler characteristic, just like the Chern number can be obtained by integrating the Berry curvature. This theorem is the celebrated Chern-Gauss-Bonnet theorem, that we state here \vspace{0.4cm}~\cite{bellGaussBonnet}:

\begin{theorem}
    Let $E$ be a real orientable Riemannian vector bundle of even rank $n=2k$ over a  compact manifold without boundary $M$ of dimension $n$. The Euler number satisfies: $$\left(\frac{1}{2 \pi}\right)^k \int_M \operatorname{Pf} (\Omega)=e_2(E)$$ where $\operatorname{Pf}$ is the Pfaffian operator, $\Omega$ is a curvature matrix obtained from any connection on $E$, and $e_2(E)$ is the Euler number.
\end{theorem}
\vspace{0.4cm}
 In other words, the integration of a local geometrical quantity, the curvature, is related to a global topological quantity, the Euler number.

In our case, we want to study degeneracies between the two principal bands, that is, we will consider the eigenspace bundle introduced in Section~\ref{annihilation} whose fiber is the direct sum of the eigenspaces associated with the two principal bands, namely 2 and 3. In that case, the bundle is of rank $n=2$ over the Brillouin zone, which is a 2-dimensional torus. 

First, let's assume that the last two bands are isolated from the first one, that is, for all $k$, $E_{1}(k)< E_{2}(k)\leq E_{3}(k)$. We can define the Euler connection associated with those two bands as the following differential form of degree 1 over the Brillouin zone: $$A_{k}=\langle u_2(k) \mid d_k u_3\rangle$$where $d_ku_3$ is the total derivative (or differential)~\cite{LeeSM} of the function $u_3:BZ\rightarrow \mathbb{R}^3$ at point $k\in BZ$. 

The Euler connection is very similar to the usual Berry connection, but there are two main differences: it concerns two eigenstates, and it is adapted for real eigenstates only. We prove in Appendix~\ref{annexeEuler} that the Euler connection defined this way is indeed a connection on the rank-2 vector bundle. 

This quantity is not gauge invariant: when performing a gauge transformation of the form $(u_2(k),u_3(k))\rightarrow (\pm u_2(k),\pm u_3(k))$, the Euler connection may change sign. More generally, when the bands 2 and 3 might be degenerate, the associated eigenspace is 2-dimensional, and the allowed gauge transformations are $2\times 2$ orthogonal matrices $X(k) \in O(2)$. If we write $u_2'(k)=X(k)u_2(k)$ and $u_3'(k)=X(k)u_3(k)$, we get the following transformation rule for the Euler connection: $$A_k^{\prime}=A_k+\left\langle X(k) u_2(k) \mid u_3(k)d_kX \right\rangle=A_k+\alpha_k$$

Now that we are equipped with the connection between two bands, which specifies the parallel transport, we define the Euler curvature, given by 
the exterior derivative of the Euler connection~\cite{tu_differential_2017}:
$$F=dA$$
It is a differential form of degree 2.
Contrary to the Euler connection, it is gauge invariant.
Indeed, taking the exterior derivative of the transformed Euler connection found earlier, one obtains: $$F^{\prime}=d A^{\prime}=F+d\alpha$$ and we must show that $d\alpha=0$. In the real inner product, the adjoint of $X(k)$ is the transpose $X(k)^\top$, so: $$\alpha_k=\left\langle X(k) u_2(k) \mid d_k X\: u_3(k)\right\rangle=\left\langle u_2(k) \mid X(k)^{\top} d_k X \:u_3(k)\right\rangle$$
Let's write explicitly $X(k)\in O(2)$ in the basis $(u_2(k),u_3(k))$ of the fiber: $$X(k)=\left(\begin{array}{cc}
\cos \theta(k) & -\sin \theta(k) \\
\sin \theta(k) & \cos \theta(k)
\end{array}\right)$$
We find: $$X(k)^{\top} d_k X=\left(\begin{array}{cc}
0 & -1 \\
1 & 0
\end{array}\right) d_k \theta=\mathbf{j}\: d_k \theta$$
where we have identified the matrix $\left(\begin{array}{cc}
0 & -1 \\
1 & 0
\end{array}\right)$ in $SU(2)$ with the quaternion $\mathbf{j}$. The appearance of quaternions in this context will be developed in more detail in Section~\ref{sec:quaternions}, and a review is given in Appendix~\ref{annexeQuaternion}.

Therefore: $$\begin{aligned}
\alpha_k=\left\langle u_2(k) \mid \mathbf{j} u_3(k)\right\rangle d_k \theta & =\left\langle u_2(k) \mid-u_2(k)\right\rangle d _k\theta \\
& =-d _k\theta=d_k(-\theta)
\end{aligned}$$
or in other words, $\alpha$ is the differential of a function. With $d^2=0$ for the exterior derivative~\cite{tu_differential_2017}, we obtain that $F'=F$.

Similarly to the Chern number, we find the Euler number of two bands by integrating the Euler curvature over the whole Brillouin zone~\cite{ahn_failure_2019}:
$$e_2=\frac{1}{2\pi}\int_{B Z} F_k$$
Note that the Brillouin zone, as a torus, is a 2-dimensional compact manifold, which makes the integration of this differential form of degree 2 possible.

\vspace{0.5cm}
However, when the two first bands are not isolated from the third one in the whole Brillouin zone, their Euler number is ill-defined, because the rank 2 vector bundle defined in Section~\ref{annihilation}  doesn't exist anymore. This is the case in our configuration in panels $(d)-(f)$ of Fig.~\ref{fig:6panel}.

For this reason, we use a related concept called the patch Euler class. Although we can't define the Euler number on the whole Brillouin zone, we can still define a version of it on a smaller region $D$ which doesn't contain adjacent Dirac points: $$e_2(D)=\frac{1}{2 \pi}\left[\int_{D} F_k-\oint_{\partial D} A_{k}\right]$$
The patch Euler number of a region quantifies the topological charge of the principal Dirac points inside this region. It can take integer values as well as half-integer values~\cite{peng_phonons_2022}.
Notice that when $D$ is taken to be the whole Brillouin zone, which doesn't have a boundary, we find the usual Euler number.

It is important to note that the patch Euler class of a region does not depend only on the Dirac points it contains, but also on the specific choice of gauge chosen. Furthermore, the sign of the patch Euler class is gauge dependent~\cite{peng_phonons_2022}.

Because the vector bundle is not trivial, it is not possible to find a global, continuous choice of eigenstate (a so-called global section) on the whole Brillouin zone. However, it is still possible to get, in some sense, an \textit{almost} global section by removing a curve from the Brillouin zone, and the section will be well-defined away from this curve. This kind of curve is called a Dirac string. Dirac strings were initially introduced in the context of electromagnetism as a mathematical construct to allow the existence of magnetic monopoles~\cite{DiracString}. When traversing a Dirac string, the section undergoes a gauge transformation.  It turns out that the extremities of the Dirac strings are exactly the Dirac points~\cite{bouhon_non-abelian_2020}. Note that the Dirac strings are not physically observable, meaning that they only appear because of the impossibility of a global gauge. 

\begin{figure}[h!]
    \centering
    \includegraphics[width=1\linewidth]{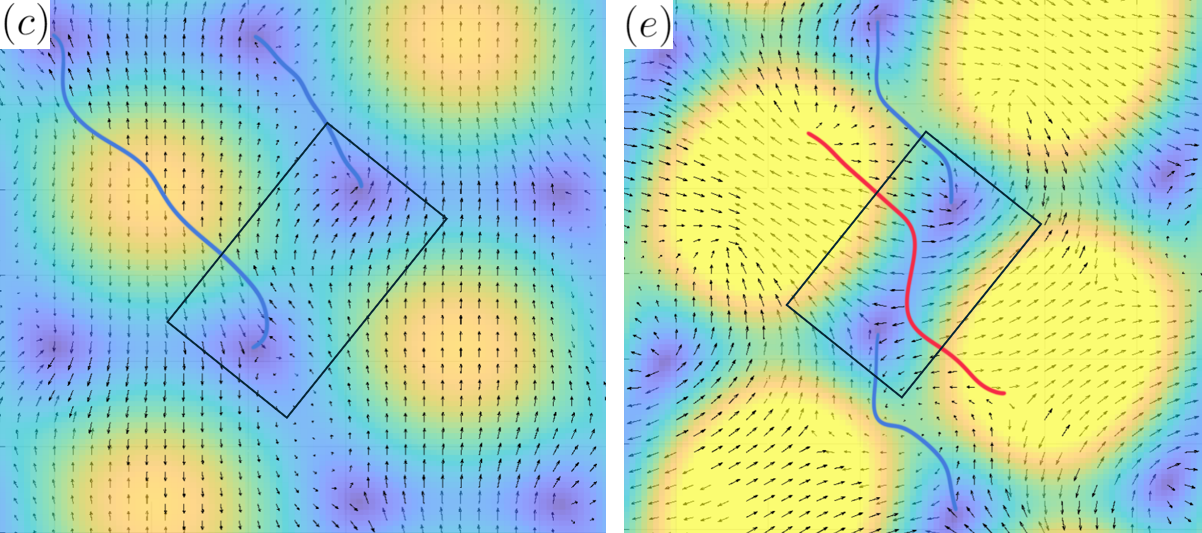}
    \caption{Choice of an eigenstate corresponding to the second band, in the configurations $(c)$ and $(e)$ of Figure~\ref{fig:6panel}. The background color still represents the difference of energies between band 2 and 3. The relevant Dirac strings have been added: in blue for the principal Dirac strings, and in red for the adjacent Dirac strings. The patch Euler number is computed in the black rectangle.}
    \label{fig:Dirac}
\end{figure}

In order to illustrate this concept, we plot in Fig.~\ref{fig:Dirac}  the eigenstate corresponding to the second band in two configurations, one where the annihilation is possible and the other where it is obstructed. We use the code provided in~\cite{bouhon_non-abelian_2020} to efficiently compute the patch Euler class of a rectangle in those two configurations. We observe that, after the tuning of the external parameters $E_A$ and $t_{BC}$, an adjacent Dirac string, with extremities two adjacent Dirac points, has appeared in panel $(e)$ and is going through the region of integration (patch). Note that this string could not exist in panel $(c)$ because, at this stage, the third band was still isolated and therefore no adjacent Dirac point was present. This adjacent Dirac string has the effect of reversing the contribution of one of the Dirac points when computing the patch Euler number, which goes from $-1$ to 0. As a consequence, the Dirac points acquire the ability to annihilate in the rectangle.

\section{Quaternionic charges}\label{sec:quaternions}

This alternative point of view, introduced in~\cite{wu_non-abelian_2019}, is based on homotopy theory. Just as the usual winding number in 2D counts how many times a path encircles a point, quaternionic charges describe a more subtle type of winding which incorporates the three-band structure of our system and the associated multidimensional rotation. 
We start with the observation that rotating around a Dirac point adds a phase on the two eigenvectors of the corresponding bands (see Appendix~\ref{annexeRotation}).

A real symmetric $3\times 3$ Hamiltonian is characterized by three orthonormal eigenvectors $u_1(k),u_2(k),u_3(k)$, which can be ordered according to increasing energies. This yields a matrix $F(k)$ in $O(3)$. However, each eigenvector is only defined up to a phase of $\pm1$. Consequently, the frame space can be described as the quotient space $O(3) / \mathbb{Z}_2^3$. The main problem to investigate this space is that $O(3)$ is not connected. To this purpose, we explain in Appendix~\ref{annexeQuaternion} that it can be replaced by a better candidate, $SU(2)/Q_8$, where $Q_{8}=\{\pm 1,\pm \mathbf{i},\pm \mathbf{j},\pm \mathbf{k}\}$ is the quaternion group of order 8.
Therefore, encircling a Dirac point is equivalent to traversing a loop in this space. It is thus natural to compute its fundamental group.
It can be demonstrated that the fundamental group of this space is given by $Q_{8}$. A proof of this fact is given in Appendix~\ref{annexeQuaternion}. Therefore, we can define the quaternionic charge of a loop $\gamma$ as the quaternion $c(\gamma)$ of the corresponding homotopy class. It is important to note that this reasoning must be conducted away from Dirac points, as this construction fails at such points (since the possible gauge transformations for each eigenvector are no longer simply $\pm 1$).

Processes which can be described by quaternionic charges are often referred to as being ``non-abelian" because $Q_8$ is not commutative, for example, $\mathbf{ij} = -\mathbf{ji} \neq \mathbf{ji}$.
What is particularly interesting is that the quaternion characterizes the type of Dirac point that is encircled by the loop~\cite{wu_non-abelian_2019}. This correspondence is proven in Appendix~\ref{annexeDiracBelt}.
\begin{itemize}
\item 1 corresponds to a loop with no Dirac points inside, or two Dirac points which can annihilate inside the loop.
\item $\pm \mathbf{i}$ corresponds to a principal Dirac point.
\item $\pm \mathbf{k}$ corresponds to an adjacent Dirac point.
\item $\pm \mathbf{j}$ corresponds to a pair of each type of Dirac points.
\item $-1$ corresponds to a pair of Dirac points of the same type and same orientation, which cannot annihilate inside the loop.
\end{itemize}


We might be tempted to define the quaternionic charge of a Dirac point as \textit{the} quaternion corresponding to any small loop circling around this Dirac point. However, this is ill-defined, as two  loops $\gamma_1$ and $\gamma_2$ going on both sides of an adjacent Dirac points are conjugated in the fundamental group by the loop $c$ circling around the adjacent Dirac point: $\gamma_1=c \gamma_2 c^{-1}$.  Therefore, they have opposite charge, as shown in Fig.~\ref{fig:conjugaison}. In the same way, reversing a loop corresponds to taking the inverse of its quaternionic charge. Note that, for clarity, orientation has been omitted in all diagrams, but all loops should be understood as having a counter-clockwise orientation.

\begin{figure}[h!]
    \centering
    \includegraphics[width=0.5\linewidth]{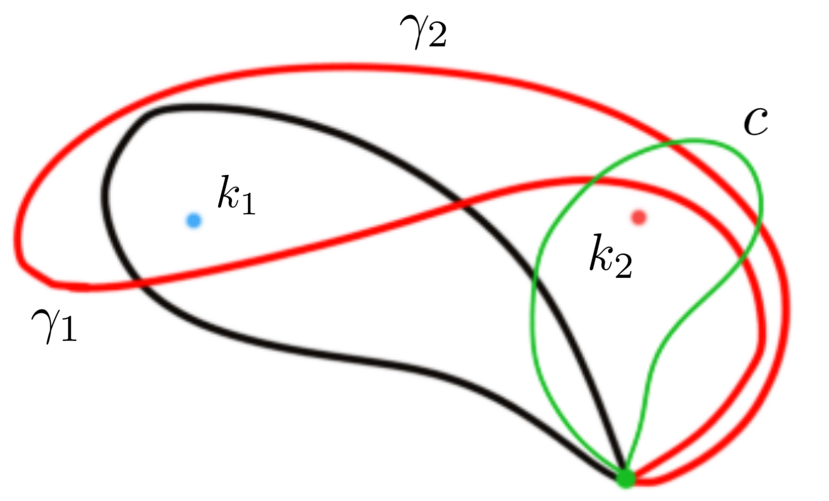}
    \caption{Two loops in reciprocal space $\gamma_1$ (in red) and $\gamma_2$ (in black) circling around the Dirac point $k_1$, with base point shown in green. They are conjugated by the loop $c$ circling around the adjacent Dirac point: $\gamma_1=c \gamma_2 c^{-1}$.}
    \label{fig:conjugaison}
\end{figure}

Therefore, it is only possible to define the quaternionic charge of a Dirac point as one of the five conjugacy classes described above.

Let's see how the quaternionic charge of a principal Dirac $k_1$ point changes when this point rotates around an adjacent Dirac point $k_2$. This process, called braiding, is shown in Fig. ~\ref{fig:rotationAutourDirac}.
\begin{figure}[h!]
    \centering
    \includegraphics[width=0.95\linewidth]{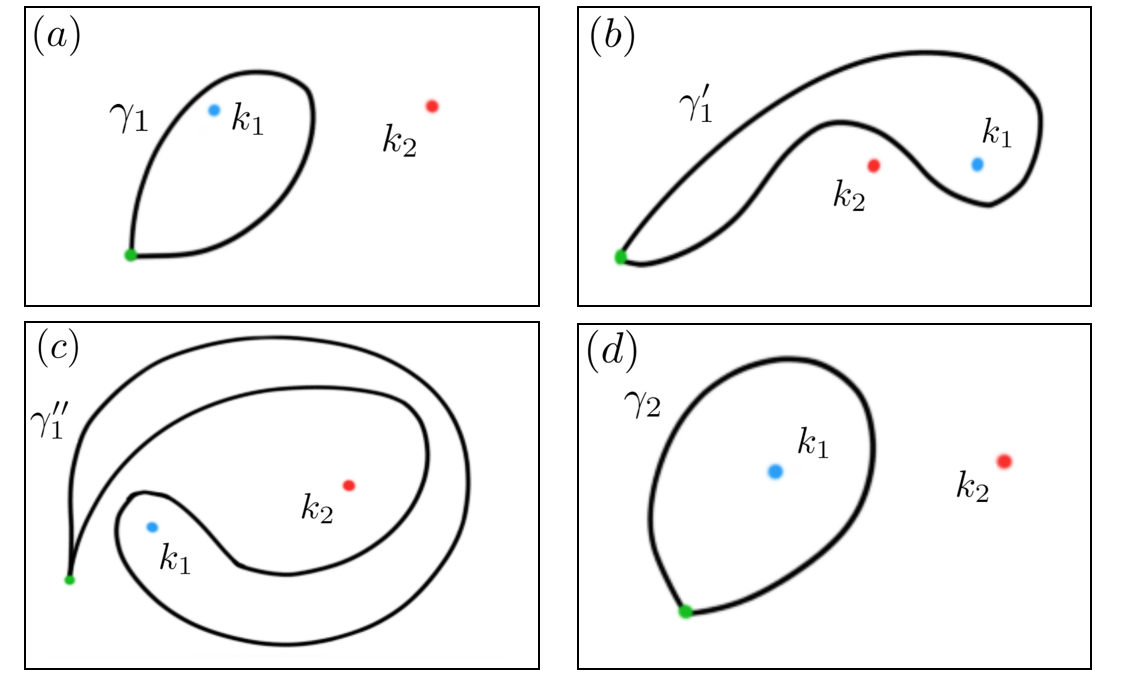}
    \caption{Braiding of a principal Dirac point $k_1$ with an adjacent Dirac point $k_2$ in reciprocal space. We follow the loop $\gamma_1$ around $k_1$, becoming successively $\gamma_1'$ and $\gamma_1''$. }
    \label{fig:rotationAutourDirac}
\end{figure}

When the braiding is complete on panel $(c)$, the configuration is back to the initial position of the Dirac points, but the rotated loop $\gamma_1''$ is conjugated with the initial loop $\gamma_1$, which means that the quaternionic charge of $k_1$ with respect to $\gamma_1$ has been reversed during the braiding, according to the argument discussed above and illustrated by Fig.~\ref{fig:conjugaison}.

Now we can prove with quaternionic charges that, in the configuration of the kagome lattice showcased in Fig.~\ref{fig:6panel}, the topological obstruction to the annihilation is removed after the braiding. This is shown in Fig.~\ref{fig:braidingkagomequaternion}: no matter the choice of initial loop $\gamma_1$ encircling the two principal Dirac points, one of the two adjacent Dirac points has to cross $\gamma_1$ during the braiding, which has the effect of reversing the orientation of the loop, according to the reasoning in Fig.~\ref{fig:conjugaison}. This shows that $\gamma_1$ and $\gamma_2$ carry an opposite quaternionic charge, which proves that braiding a principal Dirac point with an adjacent one alters their capacity to annihilate.

\begin{figure}[h!]
    \centering
    \includegraphics[width=1\linewidth]{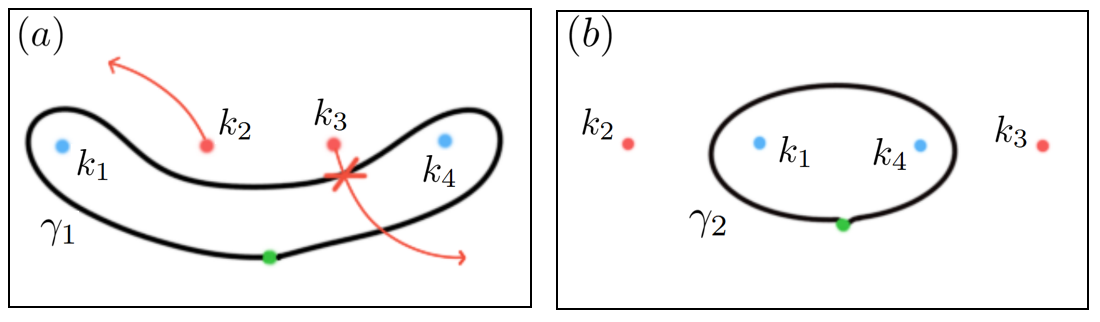}
    \caption{Braiding of two principal Dirac points (in blue) with two adjacent Dirac points (in red) in reciprocal space for the kagome lattice configuration shown in Fig.~\ref{fig:6panel}. The chosen loop around the principal Dirac points is shown in black. Red cross shows that one of the adjacent Dirac points has to cross the loop during the braiding.
    }
    \label{fig:braidingkagomequaternion}
\end{figure}

\section{Implementation in photonic systems}\label{sec6}

The approaches to the description of the behavior of the Dirac points based on the Euler number and on the quaternionic charges were suggested in previous works~\cite{lim_dirac_2020,bouhon_non-abelian_2020} which were considering rather abstract systems characterized by artificial tight-binding Hamiltonians, without going into the details of the realistic implementation of these Hamiltonians. In particular, the physical implementations of the transformations leading to the motion and annihilation of the Dirac points or the obstruction of this annihilation were not discussed.

The advantage of our model Hamiltonian describing a kagome lattice is that it can be relatively easily and with a good precision implemented in a photonic system, such as a lattice of coupled pillar microcavities~\cite{harder2021kagome} or an interference pattern created in an atomic vapor cell in the regime of electromagnetically-induced transparency (EIT)~\cite{yu2022optically}

Although the configuration we described above could be implemented directly, in this section we propose an alternative which is even easier to implement, as it is very close to the unmodified kagome lattice, and requires the modification of only one parameter, which is the energy $E_A$ of the lattice site $A$. This parameter can be modified \textit{in-situ} by using an optical pump in both suggested implementations (coupled pillar cavities~\cite{harder2021kagome} or 
EIT-lattice~\cite{yu2022optically}). This alternative is shown in Fig.~\ref{fig:kagomealternatif}.
We start with two principal Dirac points in $(a)$. They merge without annihilation in $(b)$, bounce in the orthogonal direction in $(c)$, and annihilate on the edge of the reciprocal cell in $(d)$.

\begin{figure}[h!]
    \centering
    \includegraphics[width=1\linewidth]{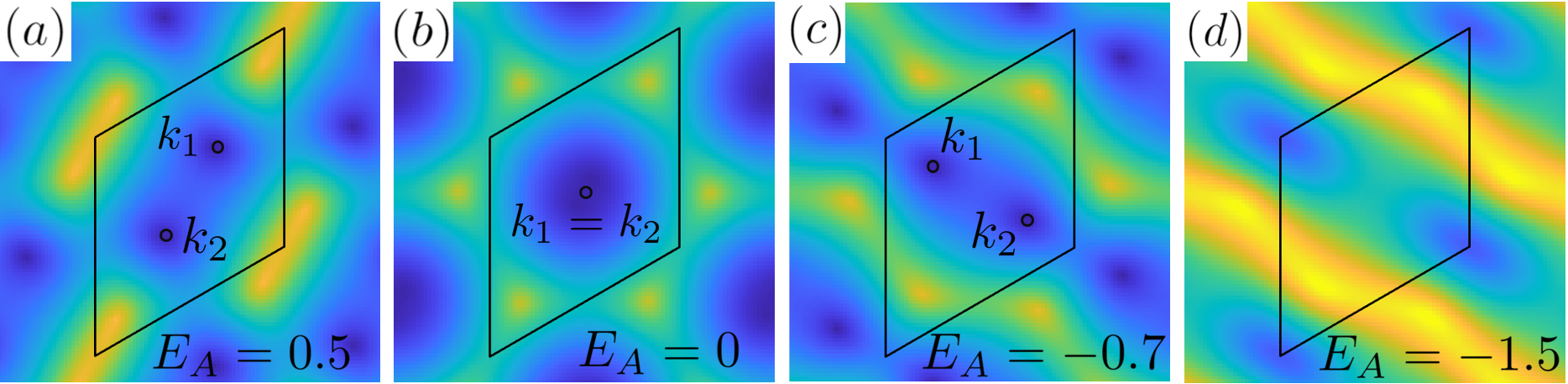}
    \caption{Alternative braiding configuration. The rhombus is a choice of a reciprocal primitive cell. In panel $(a)$, the initial configuration for the tuning parameters is $E_A=0.5$, $t_{AB}=1$, $t_{AC}=1$, $E_B=0$,  $t_{BC}=1$, $E_C=0$. The parameter $E_A$ goes from $0.5$ to $-1.5$ as shown in the bottom right corner of each panel.}
    \label{fig:kagomealternatif}
\end{figure}

The main difference is that this alternative configuration uses the periodicity of the reciprocal cell. However, the arguments presented to study the obstruction to the annihilation of Dirac points remain unchanged. Using the patch Euler number, it is a direct calculation using the algorithm from~\cite{bouhon_non-abelian_2020}, as was shown in Fig.~\ref{fig:Dirac}.

\begin{figure}[h!]
    \centering
    \includegraphics[width=0.5\linewidth]{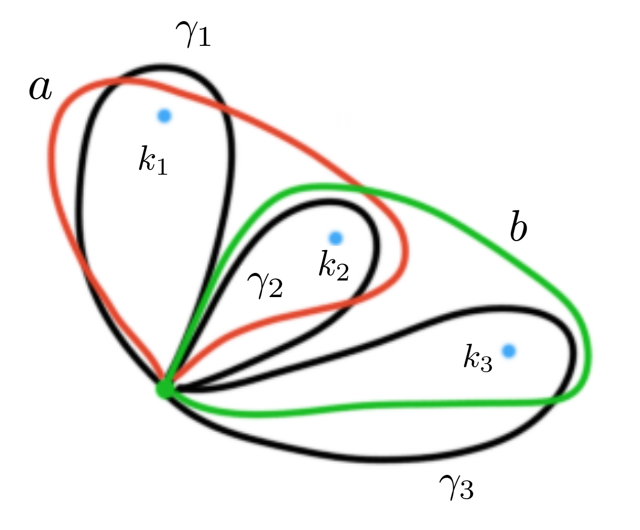}
    \caption{The three principal Dirac points in the alternative model shown in Fig.~\ref{fig:kagomealternatif}(c). The loops denoted as $\gamma_1,\gamma_2,\gamma_3$ (in black) each encircle a Dirac point. The loop denoted as $a$ (in red) and $b$ (in green) encircle two Dirac points.} 
    \label{fig:alternativeKagomeBraiding}
\end{figure}

Using quaternions, the reasoning is slightly different (Fig.~\ref{fig:alternativeKagomeBraiding}). We use three loops $\gamma_1,\gamma_2,\gamma_3$ around each of Dirac points. We define the compositions $a=\gamma_1 \cdot \gamma_2$ and $b=\gamma_2 \cdot \gamma_3$. We see that, $a \cdot \gamma_3 = \gamma_1\cdot b$. Because $k_2$ and $k_3$ are able to annihilate, the quaternionic charge $c(b)$ is 1 ; and because $\gamma_1$ and $\gamma_3$ are of opposite charge, we get $c(a)=c(\gamma_1) c(b) c(\gamma_3^{-1})=-1$, which shows that $k_1$ and $k_2$ cannot annihilate inside $\gamma_1$.
\vspace{0.5cm}

To demonstrate the possibility of experimental implementation of the full braiding protocol, we perform numerical simulations beyond the tight-binding limit, by solving the stationary Schrödinger equation  
\begin{equation}
    -\frac{\hbar^2}{2m}\Delta\psi(x,y)+\left(U_0(x,y)+U_1(x,y)\right)\psi(x,y)=E\psi(x,y)
    \label{Schr}
\end{equation}
on a grid with periodic boundary conditions of the Floquet type, accounting for the wave vector of the plane wave. This allows to find the Bloch functions for the corresponding wave vector, together with the associated eigenvalues. The equation is mapped to an eigenvalue problem for a sparse matrix solved using the ARPACK library~\cite{lehoucq1998arpack}. 
In Eq.~\eqref{Schr}, $\psi$ is the amplitude of the electric field (we restrict the consideration to a single polarization), $x,y$ are the coordinates in the plane of the photonic structure, with the wave vector in the transverse direction $k_z$ being quantized (in a microcavity) or fixed (in a vapor cell); $\hbar$ is the Planck's constant; $m$ is the effective mass of the photonic or polaritonic mode (it depends on $k_z$); $U_0(x,y)$ is the effective potential determined by the etching of the microcavity or by the interference of the EIT pump beams. This potential is constant. $U_1(x,y)$ is a variable contribution concerning only a single type of sites (e.g. the A-sites) and allowing to vary the on-site energy. It is created by optical means in all cases~\cite{yu2022optically}.

A typical result of such simulations is shown in Fig.~\ref{fig:KagNum}. Left column shows the numerically calculated dispersion $E(k_x,k_y)$ based on Eq.~\eqref{Schr} (beyond the tight-binding approximation). Right column shows the cut of the dispersion as a function of $k_x$ for $k_y=0$: the black solid line corresponds to the same numerical simulation (Eq.~\eqref{Schr}), while the red dashed curve is based on the tight-binding nearest-neighbor Hamiltonian~\eqref{Hamiltonien}. The top panel (a,b) corresponds to the reference configuration of the unperturbed kagome lattice, exhibiting a double Dirac point at the $\Gamma$ point (parabolic band touching). The bottom panel (c,d) corresponds to the perturbed case with $E_A<0$, exhibiting two Dirac points close to the $\Gamma$ point (marked by white arrows in panel (c)). Further increase of the perturbation leads to the annihilation of the DPs, whereas this annihilation is obstructed for $E_A=0$ at the $\Gamma$ point (panel (a)).

The parameters for the numerical simulation and for the tight-binding model correspond to the realistic parameters of patterned  cavities~\cite{harder2021kagome}: trap radius 1.7~$\mu$m, nearest-neighbor distance 5~$\mu$m, photonic mass $10^{-5}m_0$ ($m_0$ is the free electron mass), trapping potential depth 6~meV. The resulting tunneling coefficient for the tight-binding model is $|t|\approx 0.126$~meV (the unperturbed band width is $6t\approx 0.76$). The perturbation magnitude is $E_A=-0.12$~meV.
Using the parameters of an alternative possible implementation (EIT-based) does not lead to a qualitative change of the results~\cite{yu2022optically}, because of their topological robustness.

\begin{figure}[tbp]
    \centering
    \includegraphics[width=0.9\linewidth]{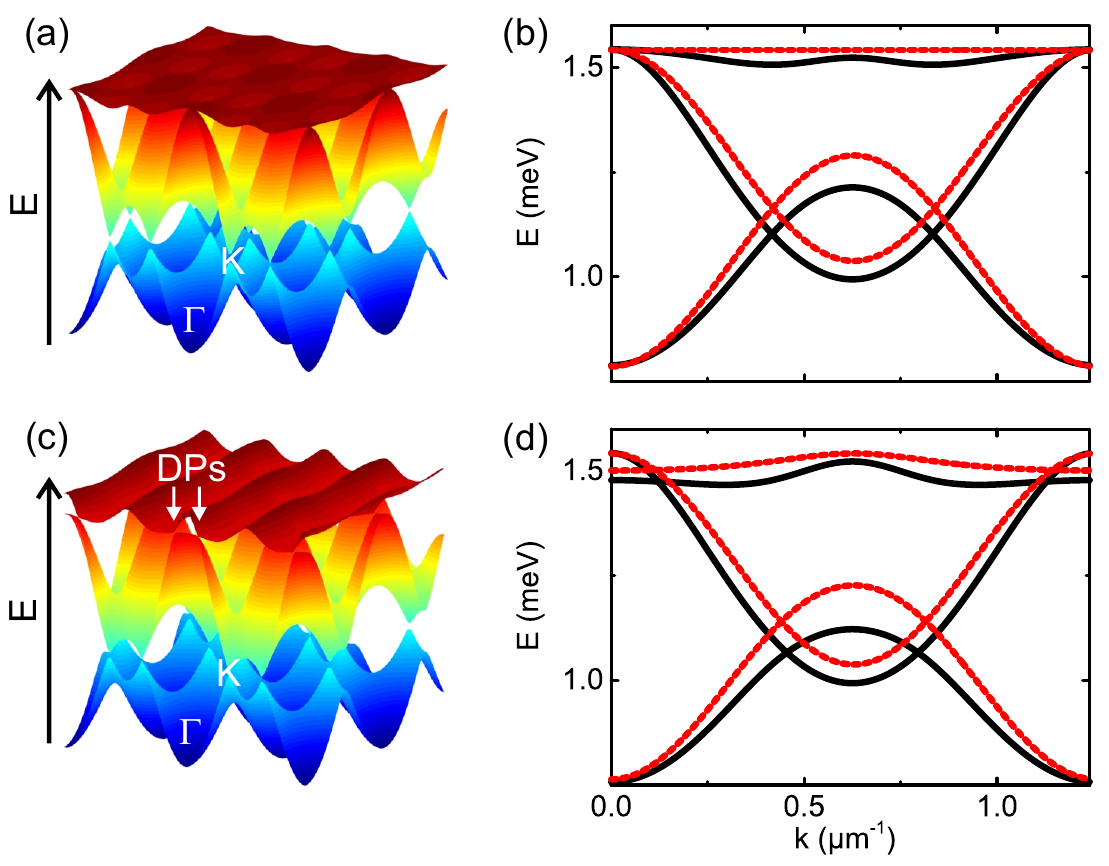}
    \caption{Dispersion of the Kagome lattice: (a,b) - unperturbed, (c,d) - perturbed by $E_A<0$. Numerically calculated beyond tight-binding (Eq.~\eqref{Schr}): (a,c) as a function of $k_x,k_y$ and (b,d) as a function of $k_x$ for $k_y=0$ (black solid line). Analytical tight-binding nearest-neighbor: (b,d) (red dashed line).}
    \label{fig:KagNum}
\end{figure}

We note that the experimental work~\cite{yu2022optically} implements exactly the configuration that we suggest for the study of the annihilation of the Dirac points in the kagome lattice, with the optical potential on the sites of a single type $E_A\neq 0$. The authors have demonstrated the possibility to create different types of potentials, including those which can be expected to lead to the annihilation of the Dirac points. However, that work did not provide any evidence of the annihilation or its obstruction. Indeed, the annihilation of the Dirac points can be expected to result in the suppression of the conical refraction~\cite{zhang2019particlelike} associated with the Dirac points~\cite{Hamilton1837,Lloyd1833}). The detection of such behavior will be the next step for the future studies.

\section{Conclusion}\label{sec7}In this work, we studied the conditions for Dirac point annihilation in a three-band kagome lattice model with time-reversal symmetry. We found that while tuning the system parameters can bring Dirac points together, their annihilation and the resulting band gap opening can be topologically obstructed. Our analysis shows that the non-triviality of a rank-2 eigenspace bundle can obstruct the annihilation, and this non-triviality is detected using the Euler class. The non-abelian braiding of Dirac points, characterized by quaternionic charges, gives a different point of view for the same question. The success of these two approaches highlights a connection between vector bundles in differential geometry and non-abelian fundamental groups in homotopy theory, and this connection is still not fully explored. Finally, we suggest a realistic protocol for experimental verification in photonics. Our results could be applied to more complex systems in future works.

\backmatter

\bmhead{Acknowledgements}

We thank J.N. Fuchs, A. Bouhon, T. Lejeune and M. Roche for valuable discussions.

\section*{Declarations}

\begin{itemize}
\item Funding: This work was supported by the European Union’s Horizon 2020 program, through a FET Open research and innovation action under the grant agreements No. 964770 (TopoLight).
 Additional support was provided by the ANR Labex GaNext (ANR-11-LABX-0014), the ANR program "Investissements d'Avenir" through the IDEX-ISITE initiative 16-IDEX-0001 (CAP 20-25), the ANR project MoirePlusPlus (ANR-23-CE09-0033), the ANR project "NEWAVE" (ANR-21-CE24-0019), and the ANR SFRI project "Graduate Track of Mathematics and Physics" (GTMP) of University Clermont Auvergne.
\item Conflict of interest: The authors declare that they have no conflict of interest
\item Ethics approval: Not applicable
\item Consent to participate: Not applicable
\item Consent for publication: Not applicable
\item Data availability: Available upon reasonable request
\item Code availability: Not applicable
\item Authors' contributions: All authors contributed to the study conception and design. All authors read and approved the final manuscript.
\end{itemize}

\begin{appendices}

\section{Construction of the kagome lattice}\label{annexeKagome}

Our purpose here is to detail the construction of the kagome lattice with the tight-binding model. Following~\cite{bena_remarks_2009}, we pay great attention to the choice of basis vectors to avoid confusion.
Let's start from a generic tight-binding Hamiltonian with multiple atoms per unit cell:
$$H=-\sum\limits_{R,\alpha,R',\beta}t_{\alpha,\beta}(R,R')|R,\alpha \rangle \langle R',\beta |$$
where $R,R'$ represent the positions of adjacent primitive cells, $\alpha,\beta$ are indices representing the different sites inside a cell, $A,B,C...$, and $t_{\alpha,\beta}(R,R')>0$ are the corresponding tunneling coefficients.

We will actually denote the cell positions as $R_{i,j}$. Note that we arbitrarily chose the cell position as the position of the cell center.
We will denote, within each primitive cell $(i,j)$, the position of atoms as $R_{ij}^A$, $R_{ij}^B$, $R_{ij}^C$.

In Fig.~\ref{fig:neighbors} we show an illustration for the nearest neighbors that each site sees.
\begin{figure}[h!]
    \centering
    \includegraphics[width=1\linewidth]{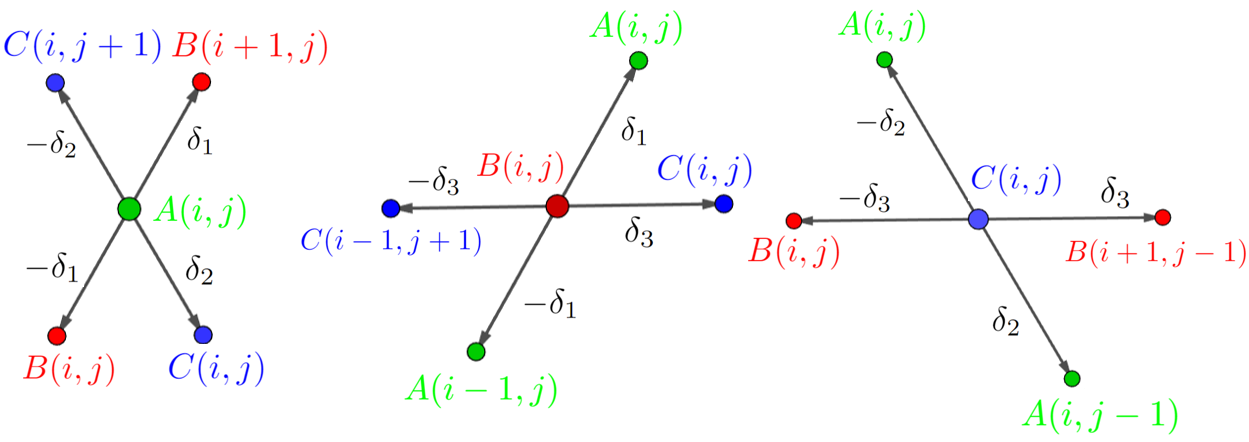}
    \caption{Nearest neighbors for each of the $A$, $B$ and $C$ sites. }
    \label{fig:neighbors}
\end{figure}

Therefore, in the simple case where the tunnel coefficient is the same for every pair of neighbors, and for the sake of simplicity $t=-1$, the tight-binding Hamiltonian is given by: $$\begin{aligned} H= & \left(\left|R_{i j}^{B}\right\rangle+\left|R_{i j}^{C}\right\rangle+\left|R_{i, j+1}^{C}\right\rangle+\left|R_{i+1, j}^{B}\right\rangle\right)\left\langle R_{i j}^{A}\right| \\ & +\left(\left|R_{i j}^{A}\right\rangle+\left|R_{i-1, j}^{A}\right\rangle+\left|R_{i-1, j-1}^{C}\right\rangle+\left|R_{i j}^{C}\right\rangle\right)\left\langle R_{i j}^{B}\right| \\ & +\left(\left|R_{i j}^{A}\right\rangle+\left|R_{i+1, j-1}^{B}\right\rangle+\left|R_{i j}^{B}\right\rangle+\left|R_{i, j-1}^{A}\right\rangle\right)\left\langle R_{i j}^{C} \right|\end{aligned}$$
\\ \vspace{0.1cm}
In the so-called base II~\cite{bena_remarks_2009}, the ``Bloch functions" are given by the Fourier transform over all cells of the orbitals at $\alpha=A,B$ or $C$: $$|k, \alpha\rangle=\frac{1}{\sqrt{N}} \sum\limits_{(i,j)}e^{i k \cdot R_{i j}^\alpha}|R_{i, j}^{\alpha}\rangle$$ (Whereas in base I, it would have been a slightly different phase:
$$|k, \alpha\rangle=\frac{1}{\sqrt{N}} \sum\limits_{(i,j)}e^{i k \cdot R_{i j}}|R_{i, j}^{\alpha}\rangle$$ 
where $R_{ij}$ is a fixed position inside the primitive cell (for example, the center).)

Actually, we should carefully call $|k,\alpha\rangle$ a Bloch function, because it is not strictly speaking an eigenstate of the Hamiltonian, although it is indeed an eigenstate of the translation operators of the crystal.

Indeed, let's examine the effect of the tight-binding Hamiltonian on, for example, $|k,A\rangle$:$$H|k, A\rangle=\frac{1}{\sqrt{N}} \sum_{(i, j)} e^{i k \cdot R_{i j}^A} H\left|R_{i j}^A\right\rangle$$
In order to find $H\left|R_{i j}^A\right\rangle$, we can just look at Fig.$~\ref{fig:neighbors}$:

$$
\begin{aligned}
H|k, A\rangle &=\frac{1}{\sqrt{N}} \sum_{(i, j)} e^{i k \cdot R_{i j}^A} 
\left(\left|R_{i j}^B\right\rangle+\left|R_{i, j+1}^C\right\rangle \right. \\
&\qquad\qquad\qquad \left. +\left|R_{i, j+1}^C\right\rangle+\left|R_{i+1, j}^B\right\rangle\right)
\end{aligned}
$$

Let's analyze one of those terms, say, the last one. Again with Fig.~\ref{fig:neighbors}, with: $$\delta_1=R_{i+1, j}^B-R_{i, j}^A$$

We find that $$\begin{aligned}
\frac{1}{\sqrt{N}} \sum_{(i, j)} e^{i k \cdot R_{i j}^A}\left|R_{i+1, j}^B\right\rangle &= \frac{1}{\sqrt{N}} e^{-i k \cdot \delta_1} \sum_{(i, j)} e^{i k \cdot R_{i+1, j}^B}\left|R_{i+1, j}^B\right\rangle
    \\&=e^{-i k \cdot \delta_1}|k, B\rangle
\end{aligned}
$$
where we used that the summation on all $i+1$ is the same as the summation on all $i$.

Reasoning in the same fashion on each term: $$
\begin{aligned}
    H|k, A\rangle&=\left(e^{-i k \cdot \delta_1}+e^{i k \cdot \delta_1}\right)|k, B\rangle+\left(e^{-i k \cdot \delta_2}+e^{i k \cdot \delta_2}\right)|k, C\rangle \\
    &= 2 \cos \left(k \cdot \delta_1\right)|k, B\rangle+2 \cos \left(k \cdot \delta_2\right)|k, C\rangle
\end{aligned}
$$

This shows that, although the $|k, A\rangle$ are not eigenstates of $H$, we have that, for fixed $k$, the Hamiltonian $H$ is stable on the subspace spanned by $|k, A\rangle,|k, B\rangle,|k, C\rangle$. Writing the restriction of $H$ in this basis gets us the to the Bloch Hamiltonian: $$H(k)=2 \left(\begin{array}{ccc}0 & \cos \left(k \cdot \delta_1\right) & \cos \left(k \cdot \delta_2\right) \\ \cos \left(k \cdot \delta_1\right) & 0 & \cos \left(k \cdot \delta_3\right) \\ \cos \left(k \cdot \delta_2\right) & \cos \left(k \cdot \delta_3\right) & 0\end{array}\right)$$
Adopting the convention that $t>0$, and factoring the minus sign out of the matrix, we get the Bloch Hamiltonian introduced in the main text.
Note that the choice of base II is necessary for the Bloch Hamiltonian $H(k)$ to be real~\cite{lim_dirac_2020}. This brings the downside that the eigenvectors do not have the periodicity of the Brillouin Zone, although the energies still do.

\section{Vector bundles, the Euler connection, and the Chern-Gauss-Bonnet theorem} \label{annexeEuler}

We define some essential notions for vector bundles. An introduction to this topic can be found in~\cite{tu_differential_2017}, and its link with solid state physics is explained in~\cite{sergeevVectorBundle2023}.
Roughly speaking, a vector bundle is a mathematical object which locally looks like a cartesian product, but not necessarily globally. More formally, a vector bundle consists in the data of a manifold $B$ called the base manifold, a manifold $E$ called the total space, and a finite-dimensional vector space $F$ called the fiber. They are related by a smooth, surjective projection map $\pi : E \rightarrow B$, satisfying the following condition of \textit{local linear trivialization}:
each point $b \in B$ has an open neighborhood $U$ and a diffeomorphism $\phi:\pi^{-1}(U)\longrightarrow U\times F$ such that:
\begin{itemize}
    \item The map $\phi$ is compatible with the projection, meaning that if $\phi(e)=(b',f)$, then $\pi(e)=b'$ for any $e \in \pi^{-1}(U)$. This can be written concisely as $\mathrm{proj}_1 \circ \phi = \pi$.
    \item For each point $b' \in U$, the restriction of $\phi$ to the fiber over $b'$, denoted $\phi_{b'} : \pi^{-1}(\{b'\}) \to \{b'\} \times F$, is an isomorphism of vector spaces.
\end{itemize}
This data is often summarized by the sequence $F \rightarrow E \xrightarrow{\pi} B$.

A simple example of vector bundle is the (infinite) cylinder $S^1 \times  \mathbb{R}$. It is a global cartesian product of the circle (the base space) and $\mathbb{R}$ (the fiber), so it is trivial.
An example of a non-trivial vector bundle is the Möbius vector bundle~\cite{sergeevVectorBundle2023}. Although it is also a bundle with base space the circle, and with fiber $\mathbb{R}$, it cannot be written globally as a cartesian product, because it is not orientable.

Another classical example of a vector bundle is the tangent bundle of a manifold: it consists of all the tangent vectors at all points of the manifold.

For our purposes, we investigate eigenspaces bundles. An eigenspace bundle is a vector bundle with base space the Brillouin zone, and the fiber at point $k$ is the direct sum of any number of eigenspaces of the Bloch Hamiltonian $H(k)$ (equation (~\ref{Hamiltonien})). It is real, because $H(k)$ is real symmetric.
For example, if we choose to attach all the eigenstates, we get a rank-3 vector bundle. This particular eigenspace bundle is the cartesian product $BZ\times \mathbb{R}^3$ and is therefore trivial.  It is nonetheless interesting, as it will allow us to define sub-bundles which are non trivial.
 \begin{figure}[h!]
    \centering
    \includegraphics[width=1\linewidth]{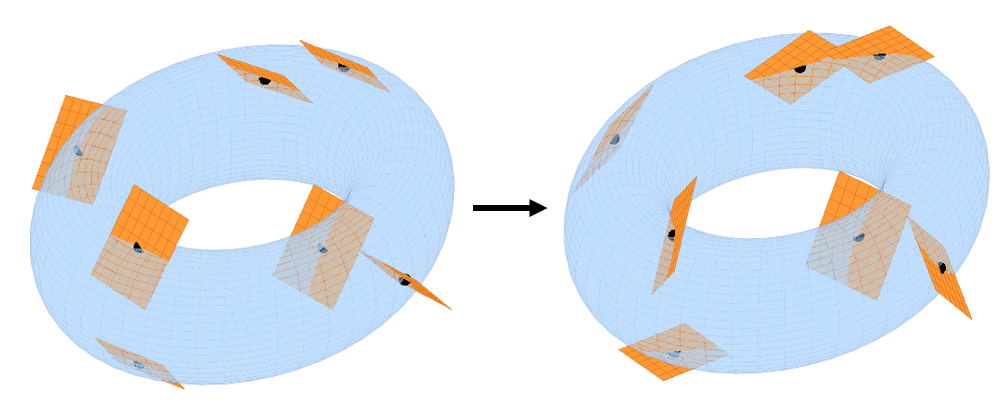}
    \caption{The rank-2 eigenspace bundle undergoes deformations as the external parameters of $H(k)$ are modified. Here, the parameter $a$ goes from $0$ to $1$. }
    \label{fig:deformbundle}
\end{figure}
The main vector bundle of interest in this article is the one obtained by considering only the two eigenspaces corresponding to the two principal bands, hence a rank-2 bundle such as in Fig.~\ref{fig:deformbundle}. Note that this vector bundle is only well-defined when the third band is isolated from the two principal bands. 

To give a more technical construction of this bundle, note that for a Hamiltonian $H(k)$, there is a continuous map $f:T^2 \rightarrow \operatorname{Gr}(2,3)$ from the Brillouin zone to the Grassmannian of subspaces of dimension 2 of $\mathbb{R}^3$, defined as $f(k)$ being the direct sum of the eigenspaces of consideration. There is a natural vector bundle on the Grassmannian, called tautological bundle, whose fiber at each point is the subspace corresponding to this point~\cite{milnor_characteristic_nodate}. The eigenspace bundle is then defined as the pullback bundle of the tautological bundle by $f$~\cite{sergeevVectorBundle2023}.

One might wonder if the modification of the external parameters of the Hamiltonian could induce a topological change in the eigenspace bundle. It turns out that this is not the case, as long as the third band stays isolated.
To see why, we use the following trick~\cite{husemoller1994fibre}: instead of considering, for each parameter $t$, the eigenspace bundle $E_t$ associated with $H_t (k)$ with base space the Brillouin zone $B=BZ$, we can directly consider the eigenspace bundle $E$ with base space $B \times I$ where the parameter $t$ evolves in $I$. The restrictions $$E_{\mid B  \times\{0\}}, E_{B  \times\{1\}}$$ are precisely $E_0$ and $E_1$. With those notations, we have the following~\cite{husemoller1994fibre}:
\vspace{0.4cm}

\begin{theorem}
    If $E$ is a vector bundle with base $B\times I$, with $B$ a paracompact, Hausdorff space, then we have the vector bundle isomorphism: $$E_{\mid B  \times\{0\}}\cong E_{\mid B  \times\{1\}}$$
    
\end{theorem}

Here the Brillouin zone $B=BZ$ is a torus, which is compact (and it particular paracompact) and Hausdorff, so this general theorem applies in our case. It shows that the eigenspace bundle remains topologically the same when the external parameters are modified.

\vspace{1cm}
Subsequently, we want to study the rank-2 eigenspace bundle by applying the Chern-Gauss-Bonnet theorem mentioned in Section~\ref{Eulernumber}. Let's verify the hypothesis needed.

A vector bundle $E\rightarrow B$ is said to be Riemannian if each of its fiber is equipped with an inner product, also called Riemannian metric. For the 2-band bundle, the inner product is taken to be the standard one on $\mathbb{R}^3$.

A connection $\nabla$ on a vector bundle is a tool which allows to generalize the notion of derivative for sections, and define the parallel transport. It is defined as the data, for each vector field $X$ of the base manifold, of an operator $$\nabla_X : \Gamma(E) \rightarrow \Gamma(E)$$ taking sections as input and output, such that $\nabla_X(s)$ is linear in $s$, $C^\infty(M)$-linear in $X$, and verifies the Leibniz rule: $$\nabla_X(f s)=f \nabla_X(s)+X(f) s$$
This definition is directly inspired by the properties of the directional derivative in Euclidean space. 

However, it is cumbersome for practical cases. Luckily, it is possible to describe its behaviour in a frame.
We choose a local frame, which amounts to choosing two eigenstates $u_2(k),u_3(k)$ locally ; the connection matrix $\omega$ in this local frame is the matrix with coefficients being differential forms of degree 1, such that: $$\begin{aligned} & \nabla_X( u_2)=\omega_{22}(X) u_2+\omega_{23}(X) u_3 \\ & \nabla_{X} (u_3)=\omega_{32}(X) u_2+\omega_{33}(X) u_3\end{aligned}$$
Subsequently, by linearity and the Leibniz rule, it is possible to compute the connection of any section in this frame.

In the case of an orthonormal frame, which is possible for us because the Hamiltonian is real symmetric (hence its eigenspaces are orthogonal), when the connection is compatible with the Riemannian metric, the connection matrix is skew-symmetric. 

In our situation, we can create a natural connection by specifying the following skew-symmetric connection matrix:$$\omega = \left(\begin{array}{cc}0 & \omega_{23} \\ -\omega_{23} & 0\end{array}\right)$$where $\omega_{23}=A_k =\langle u_2(k) \mid d_k u_3\rangle$. It is common to refer to the differential 1-form $A_k$ as \textit{the} Euler connection, and this convention was adopted in the main text, but one should keep in mind that it is only \textit{one} coefficient of the connection matrix in one specific local frame.

Following this, as with the connection matrix, we define the curvature matrix $\Omega$ relative to the frame. It is possible to prove the so-called Cartan's second structure equation, which gives a link between the connection and curvature matrices: $$\Omega=d\omega + \omega\wedge\omega$$ where the exterior derivative $d$ acts on each coefficient, and the wedge product $\wedge$ is a matrix product of differential forms.

Chern noticed that to change the frame from $(u_2(k),u_3(k))$ to $(u'_2(k),u'_3(k))$ with a gauge transformation $Q:U \rightarrow GL_n(\mathbb{R})$, the transformation formula for the curvature matrix is the usual conjugation: $\Omega' = Q\Omega Q^{-1}$. Later, he had the idea to apply a polynomial $P$ which is invariant under conjugation, such that $P(\Omega)=P(\Omega')$. For our case, the gauge transformations are elements of $O(2)$, and it can be demonstrated that the Pfaffian is an $O(2)$-invariant polynomial. 

The idea is that a skew-symmetric matrix has its determinant which is a perfect square ; the Pfaffian of a skew-symmetric matrix is defined as the square root of the determinant, which is a polynomial in the coefficients of the matrix.

In the end, we obtain a differential 2-form $\operatorname{Pf}(\Omega)$ which is globally defined on the base manifold:$$\operatorname{Pf}(\Omega)=\operatorname{Pf}\left(\begin{array}{ll}0 & dA \\ -dA & 0\end{array}\right)=\sqrt{\det\left(\begin{array}{ll}0 & dA \\ -dA & 0\end{array}\right)}=dA=F$$It is possible to prove that it is a closed form, meaning that its exterior derivative is 0. We can therefore take its cohomology class $[\operatorname{Pf(\Omega)}]\in H^2(T^2;\mathbb{R})$. Normalizing it, we obtain the Euler class $\frac{1}{2\pi}\operatorname{Pf(\Omega)}$, and the content of the Chern-Gauss-Bonnet theorem is that the integral of the Euler class on the base manifold gives the Euler number.

\section{Time-reversal symmetry in vector bundles}\label{annexeTimereversal}
In this appendix, we demonstrate that the kagome lattice exhibits time-reversal symmetry~\cite{Ludwig2016}, and we examine how this ensures that the rank-1 eigenspace bundles defined in Section~\ref{annihilation} are trivial. Note that this specific symmetry condition is not strictly necessary, as the system considered in previous work~\cite{bouhon_non-abelian_2020} possessed $C_{2}\mathcal{T}$-symmetry instead.

In spinless systems, the time-reversal operator $T$ of a Hilbert space $H$ is defined as an operator $T:H \rightarrow H$ which preserves the position operator and reverses the momentum operator~\cite{Zirnbauer2021}: $T^{-1}\hat{x}T=\hat{x},\quad T^{-1}\hat{p}T=-\hat{p}$. $T$ is antiunitary, and can be generally written as $T=UK$, with $U$ a unitary operator and $K$ the complex conjugation operator.
A system is time-reversal symmetric if the Hamiltonian $H$ commutes with the time-reversal operator $T$: $HT=TH$.
In the case of a tight-binding model, this symmetry implies that the corresponding Bloch Hamiltonian $H(k)$ satisfies the condition~\cite{Ludwig2016}: $$H(k)^*=H(-k)$$where the star denotes complex conjugation. In this article, we defined the Bloch Hamiltonian of the kagome lattice in equation ~\eqref{Hamiltonien}: $$H(k)=\begin{pmatrix} 
E_A & t_{AB}\cos(k \cdot \delta_1) & t_{AC}\cos(k \cdot \delta_2) \\ 
t_{AB}\cos(k \cdot \delta_1) & E_B & t_{BC}\cos(k \cdot \delta_3) \\ 
t_{AC}\cos(k \cdot \delta_2) & t_{BC}\cos(k \cdot \delta_3) & E_C
\end{pmatrix}$$
Because it is a real matrix, the condition becomes $H(k)=H(-k)$, which is verified with any value of the parameters, because cosine is an even function. 

We now prove that a rank-1 real eigenspace bundle with time-reversal symmetry is trivial. 
As established in Appendix~\ref{annexeEuler}, a rank-1 eigenspace bundle is the pullback of the tautological bundle $\gamma^1$ by a continuous map $f:T^2 \rightarrow \mathbb{R}P^2$, where $\mathbb{R}P^2= \operatorname{Gr}(1,3)$ is the projective space, and the Brillouin zone $BZ$ has been identified with the torus $T^2 \cong S^1\times S^1$. 
A rank-1 vector bundle is trivial if and only if it is orientable~\cite{husemoller1994fibre}. The orientability of real vector bundles can be detected using the first Stiefel-Whitney class $w_1$~\cite{milnor_characteristic_nodate}. So, we want to compute $w_1(E)$.
We have, by the naturality property of characteristic classes ~\cite{milnor_characteristic_nodate}: $$w_1(E)=w_1(f^* \gamma^1)=f^*(w_{1}(\gamma^1))$$where the second $f^*$ is the induced map on cohomology groups: $$f^* : H^1\left(\mathbb{R} P^2, \mathbb{Z} / 2 \mathbb{Z}\right)\rightarrow H^1\left(T^2, \mathbb{Z} / 2 \mathbb{Z}\right)$$
It is well-known~\cite{HatcherAlgTop} that those groups are: $$H^1(\mathbb{R} P^2, \mathbb{Z} / 2 \mathbb{Z})\cong \mathbb{Z} / 2 \mathbb{Z},\:\: H^1(T^2, \mathbb{Z} / 2 \mathbb{Z}) \cong (\mathbb{Z} / 2 \mathbb{Z})^2$$

Because the tautological bundle is not orientable~\cite{husemoller1994fibre}, its first Stiefel-Whitney class $w_1\left(\gamma^1\right)$ is non-zero in $\mathbb{Z} / 2 \mathbb{Z}$. A priori, $f^*\left(w_{1}\left(\gamma^1\right)\right)$ could be any of the four elements in $(\mathbb{Z} / 2 \mathbb{Z})^2$.

However, the time-reversal symmetry condition $H(k)=H(-k)$ implies that the eigenspace at momentum $k$ is identical to the one at $-k$. Consequently, the map $f:T^2 \longrightarrow \mathbb{R}P^2$ satisfies $f(k)=f(-k)$. This property allows $f$ to be factored through the quotient space $T^2/\sim$, where the equivalence relation $\sim$ on the torus is defined as follows:  $$k_1 \sim k_2\Leftrightarrow k_1 = k_2 \;\text{or}\; k_1=-k_2$$ By the universal property of quotient spaces~\cite{LeeTM}, because $f:T^2 \longrightarrow \mathbb{R}P^2$ is compatible with this equivalence relation, there exists a map $\tilde{f}:T^2\!/\!\sim \:\longrightarrow \mathbb{R}P^2$ such that the following diagram commutes: 
\begin{center}
    \begin{tikzcd}
T^2\!/\!\sim \arrow[rd, "\widetilde{f}"] &               \\
T^2 \arrow[r, "f"'] \arrow[u]          & \mathbb{R}P^2
\end{tikzcd}
\end{center}
We now prove that that the quotient space $T^2 / \sim$ is homeomorphic to the sphere $S^2$.

The torus is homeomorphic to a quotient of $\mathbb{R}^2$ by $\mathbb{Z}^2$. A fundamental domain is given by the square $[0,1]\times [0,1]$, as shown on panel $(a)$ of Fig.~\ref{fig:cylindresphère}. The additional identification $(x,y)\sim (-x,-y)$ imposed by the time-reversal symmetry translates, in the square, as the central symmetry $(x,y)\sim(1-x,1-y)$. This can be seen in the following way: $(x,y)$ is identified with the point $(-x,-y)$ which is outside of the fundamental domain ; by adding $(1,1)$, we get the point $(1-x,1-y)$ in the fundamental domain. 

In turn, this is equivalent to considering just half the square $[0,1]\times [0,\frac{1}{2}]$ (panel $(b)$ of Fig.~\ref{fig:cylindresphère}) with the identifications $$(0,y)\sim (1,y),  (x,0)\sim (1-x,0)\text{ and }(x,\frac{1}{2})\sim (1-x,\frac{1}{2})$$
The first identification rolls the square into a cylinder. The second and third identifications collapse the top and bottom circular boundaries of the cylinder, yielding a space homeomorphic to a sphere~\cite{HatcherAlgTop}, as illustrated on Fig.~\ref{fig:cylindresphère}.

\begin{figure}[h!]
    \centering
    \includegraphics[width=1\linewidth]{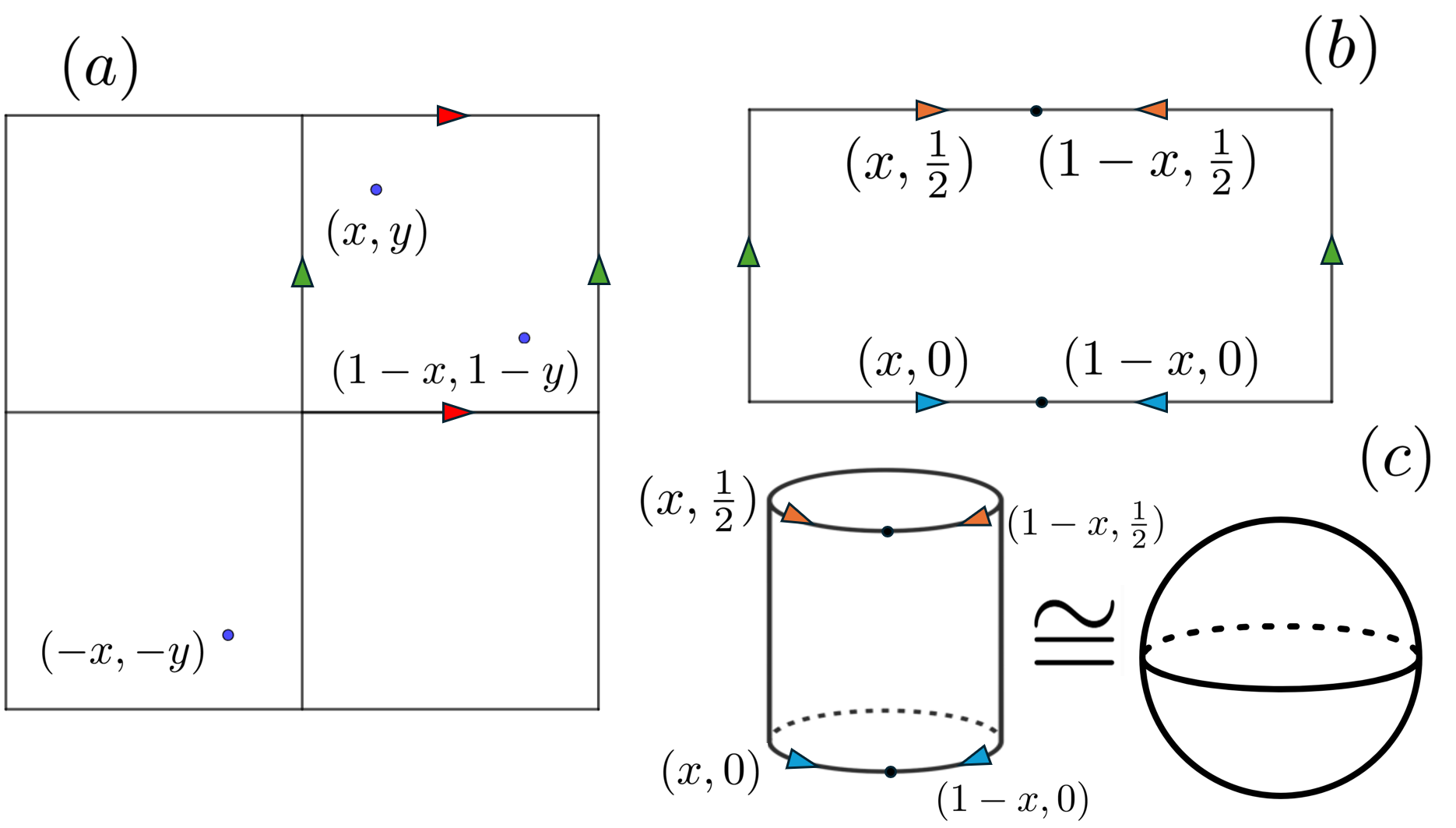}
    \caption{Visualisation of the quotient space $T^2/\sim \cong S^2$. Identifications of edges are denoted with triangles of the same color. $(a)$ Time-reversal symmetry on the fundamental domain of the torus (top-right square) is a central symmetry. $(b)$ The quotient space is generated by only an half-square.
    $(c)$ Cylinder obtained by gluing the segments marked with green triangles in the half-square. When its circular boundaries are collapsed (orange and blue triangles), we obtain the sphere $S^2$.
    }
    \label{fig:cylindresphère}
\end{figure}

It is well-known that the first cohomology group of the sphere is trivial~\cite{HatcherAlgTop}: $$H^1(S^2 ,\mathbb{Z}/2\mathbb{Z})\cong 0$$
This factorization of maps induces a corresponding diagram in cohomology. We get that $f^*$ factors by the zero map, which shows it is itself zero, as shown in the induced commutative diagram (with arrows reversed, given that cohomology is a contravariant functor~\cite{LeeSM}): 
\begin{center}
\begin{tikzcd}
0 \arrow[d, "0"']            &                                                                    \\
(\mathbb{Z}/2\mathbb{Z})^2  & \mathbb{Z}/2\mathbb{Z} \arrow[lu, "\tilde{f}^*"'] \arrow[l, "f^*"]
\end{tikzcd}
\end{center}
Finally, we have that the first Stiefel-Whitney class of our bundle is $$w_1(E)=f^*(w_{1}(\gamma^1))=0$$ which proves that the bundle is orientable, and thus trivial because it is of rank 1, which concludes the proof.

\section{Rotation around a Dirac point}\label{secA2}\label{annexeRotation}

We show that rotating around a Dirac point (say, associated with bands 2 and 3) results in an added phase factor $-1$ for the two corresponding eigenstates, and no added phase for the remaining eigenstate~\cite{wu_non-abelian_2019}:$$(u_1(k),u_2(k),u_3(k))\longrightarrow (u_1(k),-u_2(k),-u_3(k))$$
Without loss of generality, we can study a small region containing only a principal Dirac point at $k=0$. The starting point is the Löwdin method. A pedagogical presentation can be found in ~\cite{lim_dirac_2020}.
We have that the time-independant Schrödinger equation $H\psi = E\psi$ is equivalent to the two following equations: $$\left\{\begin{array}{c}
\left(H_{\alpha \alpha}+H_{\alpha \beta}\left(E I-H_{\beta \beta}\right)^{-1} H_{\beta \alpha}\right) \psi_\alpha=E \psi_\alpha \\
\psi_\beta=\left(E I-H_{\beta \beta}\right)^{-1} H_{\beta \alpha} \psi_\alpha
\end{array}\right.$$ where we decomposed $H$ along the direct sum $A\oplus B$ where $A$ is the direct sum of the eigenspaces associated with the bands 2 and 3, and $B$ the eigenspace associated with the band 1. Note that the first equation is in terms of a so-called effective Hamiltonian $H_{\alpha \alpha}+H_{\alpha \beta}\left(E I-H_{\beta \beta}\right)^{-1} H_{\beta \alpha}$

This effective $2\times 2$ Hamiltonian is real, and so it can be written in terms of the Pauli matrices, at first order, near 0:$$H(k)=h_0 I_2+h_x(k)\sigma_x+h_z(k)\sigma_z$$ where $h_0$ is a constant and $h_x$, $h_z$ are linear. The eigenstates of this Hamiltonian are~\cite{sergeevVectorBundle2023}: $$\binom{\cos (\theta / 2)}{ \sin (\theta / 2)},\binom{- \sin (\theta / 2)}{\cos (\theta / 2)}$$
\vspace{0.1cm}\\ But when we rotate around the Dirac point, the angle of $k$ acquires an additional $2\pi$  and, by linearity, $h(k)$ also acquires an angle of $2\pi$, which corresponds to changing $\theta$ to $\theta+2\pi$, and we deduce that the eigenstates transform as: $$\binom{\cos (\theta/2+\pi)}{ \sin (\theta / 2+\pi)}=-\binom{\cos (\theta/2)}{ \sin (\theta / 2)},\binom{\sin (\theta / 2+\pi)}{\cos (\theta / 2+\pi)}=-\binom{\sin (\theta / 2)}{\cos (\theta / 2)}$$which in turn shows that the two corresponding eigenstates of the form $\psi = (\psi_\alpha,\psi_\beta)$ of the full Hamiltonian acquire a phase factor $- 1$.

By reversing the roles of $A$ and $B$ in this argument, we can obtain a $1\times 1$ effective Hamiltonian associated with $B$, whose eigenstate is constant with respect to  $k$. This shows that the eigenvector associated with the first band of the full Hamiltonian does not acquire any new phase after the rotation around the Dirac point.

\section{Quaternions and the frame space}\label{annexeQuaternion}

Quaternions are a generalization of complex numbers, forming a 4-dimensional division algebra $\mathbb{H}$ over the real numbers~\cite{quaternions}: a quaternion $q$ can be written as $q=a+b\mathbf{i}+c\mathbf{j}+d\mathbf{k}$ where $a,b,c,d$ are real numbers, and $\mathbf{i},\mathbf{j},\mathbf{k}$ are imaginary units satisfying the fundamental relations $\mathbf{i}^2 = \mathbf{j}^2 = \mathbf{k}^2 = \mathbf{ijk} = -1$. These relations lead to the non-commutative multiplication rules $\mathbf{i}\mathbf{j}=-\mathbf{j}\mathbf{i}=\mathbf{k}$, $\mathbf{jk}=-\mathbf{kj}=\mathbf{i}$ and $\mathbf{ki}=-\mathbf{ik} = \mathbf{j}$. In physics, quaternions provide an elegant way to represent rotations in three-dimensional space, offering advantages over rotation matrices by being more concise and avoiding certain phenomena like gimbal lock~\cite{quaternions}.
The quaternion group $Q_8=\{\pm1,\pm \mathbf{i},\pm \mathbf{j},\pm \mathbf{k}\}$ consists of the eight unit quaternions that form a non-abelian group under multiplication. 

We now prove that the fundamental group of the non-abelian frame space defined in Section~\ref{sec:quaternions} is given by the quaternion group $Q_8$. We review the argument given in~\cite{wu_non-abelian_2019}. For the purposes of the proof, it is preferable to use the formalism of $SU(2)$ rather than that of $O(3)$. The relevant facts about the quotient topology can be found in~\cite{LeeTM}.

The frame space is the topological quotient $O(3)/A$, where $A \cong (\mathbb{Z}/2\mathbb{Z})^3$ is the group of $3\times 3$ diagonal matrices with coefficients $\pm 1$, acting on $O(3)$ by left multiplication. To compute its fundamental group, we must first ensure it is a well-behaved topological space, namely Hausdorff and path-connected. Since $A$ is a finite group (and therefore compact), acting on the Hausdorff space $O(3)$, the resulting quotient space is Hausdorff~\cite{LeeTM}. 
Although $O(3)$ is not path-connected, its quotient $O(3)/A$ is. Indeed, for any matrix $O \in O(3)$ with $\det(O)=-1$, we can choose a matrix $D \in A$ with $\det(D)=-1$ (for example, $D=\text{diag}(-1,1,1)$). Then the matrix $O' = D \cdot O$ has $\det(O')=1$ and belongs to the same orbit as $O$. Therefore, every orbit in $O(3)/A$ contains a representative from the path-connected component $SO(3)$. Since the quotient map restricted to $SO(3)$ is continuous and surjective onto $O(3)/A$, the quotient space is path-connected.

Let's reformulate the quotient in terms of $SO(3)$. We define $B$ the subgroup of matrices of $A$ with an even number of -1 coefficients, in other words, it is the subgroup generated by the three $\pi$-rotations.
We can show the homeomorphism $O(3)/A\cong SO(3)/B$. Indeed, we restrict the quotient map to $SO(3)$: $$\tilde{\pi}: S O(3) \longrightarrow O(3) /A$$Now, acting on $SO(3)$ with $B$, we notice that this action is nullified by the quotient by $A$. This shows that the map can be taken to the quotient:$$\phi : SO(3) / B \longrightarrow O(3) / A$$and it is straightforward to prove that this map is a homeomorphism. 

The advantage of this formulation for the quotient is that we are now able to use the powerful double cover~\cite{hallRep}: $$\operatorname{Ad} : SU(2) \rightarrow SO(3) $$ Because it is a double cover, the reciprocal image $Q$ of $B$ by this map is a group containing 8 matrices. 
In order to identify what group it is exactly, it is useful to identify $SU(2)$ with the unitary quaternions
 via the mapping $$\begin{aligned} \operatorname{SU}(2)  \longrightarrow & \quad\mathbb{H} \\ \left(\begin{array}{cc}a+b i & c+d i \\ -c+d i & a-b i\end{array}\right)  \mapsto & \:a+b \mathbf{i}+c \mathbf{j}+d \mathbf{k}\end{aligned}$$
 
 In this view, the covering map associates each unitary quaternion $q=\exp(\theta u)$ with the rotation of angle $2\theta$ and axis $u$, where the imaginary quaternion $u$ is identified with a vector of $\mathbb{R}^3$.
In this light, the reciprocal image of the $\pi$-rotation along the $x$-axis is given by $\{+\mathbf{i},-\mathbf{i}\}$ because $\mathbf{i} = \exp(\mathbf{i}\pi/2)$. The reciprocal image of the $\pi$-rotation along the $y$-axis is $\{+\mathbf{j},-\mathbf{j}\}$ because $\mathbf{j}=\exp(\mathbf{j}\pi/2)$, and in the same way, the reciprocal image of the $\pi$-rotation along the $z$-axis is $\{+\mathbf{k},-\mathbf{k}\}$. Finally, the reciprocal image of the identity is $\{+1,-1\}$.
In total, we obtain that $Q=Q_8=\{\pm1,\pm \mathbf{i}, \pm \mathbf{j}, \pm \mathbf{k}\}$ is indeed the quaternion group.

By a similar reasoning as earlier, we find that $$S O(3) / B\cong SU(2) / Q_8 $$But now, we may use the following standard theorem~\cite{LeeTM}:\vspace{0.4cm}
\begin{theorem}
Let $X$ be a simply connected space, and $G$ a group acting on $X$ in a free and properly discontinuous way. The fundamental group of the quotient space $X/G$ is given by: $$\pi_{1} (X / G) \cong G$$    
\end{theorem}
\vspace{0.2cm}

In our case, we have that $SU(2)\cong S^3$ is simply connected, and the action of $Q_8$ on it is discrete, so it is properly discontinuous, as well as being free. Therefore:  $$\pi_1(S U(2) / Q_8) \cong Q_8$$which concludes the proof.

 It is possible to prove $~\cite{bouhon_non-abelian_2020}$ that this formulation is precisely equivalent to the Euler number formulation.

\section{Types of quaternionic charges}
\label{annexeDiracBelt}

In Section~\ref{sec:quaternions}, we claimed that the quaternion charge of a loop characterized the type of Dirac point encircled by the loop~\cite{wu_non-abelian_2019}. We write those rules again here:
\begin{itemize}
\item 1 corresponds to a loop with no Dirac points inside, or two Dirac points which can annihilate inside the loop.
\item $\pm \mathbf{i}$ corresponds to a principal Dirac point.
\item $\pm \mathbf{k}$ corresponds to an adjacent Dirac point.
\item $\pm \mathbf{j}$ corresponds to a pair of each type of Dirac points.
\item $-1$ corresponds to a pair of Dirac points of the same type and which cannot annihilate inside the loop.
\end{itemize}

To prove this correspondence, let's first note that it is crucial to distinguish between two different kinds of loops: a loop in the Brillouin zone, $$\gamma:S^1 \rightarrow T^2$$
and the induced loop in the quotiented space of frames $$\gamma : S^1 \rightarrow O(3)/(\mathbb{Z}/2\mathbb{Z})^3$$
Although a loop on the Brillouin zone might appear as homotopically non-trivial (if it encircles some Dirac points), the corresponding loop in the (quotiented) frame space may be trivial.

To start the proof, let's start with the most fundamental case of a loop of the Brillouin zone which circles around a single principal Dirac point. As we have seen in Appendix~\ref{annexeRotation}, a frame $(u_1(k),u_2(k),u_3(k))$ is transformed in $(u_1(k),-u_2(k),-u_3(k))$ after a rotation around this Dirac point. If we fix a basis $(u_1(k),u_2(k),u_3(k))$, this corresponds to a path in $O(3)$ between the matrices $$\left(\begin{array}{lll}1 & 0 & 0 \\ 0 & 1 & 0 \\ 0 & 0 & 1\end{array}\right)\:\text{and}\:\left(\begin{array}{ccc}1 & 0 & 0 \\ 0 & -1 & 0 \\ 0 & 0 & -1\end{array}\right)$$
This transformation is a rotation of angle $\pi$ around the unit vector $u_1(k)$, with coordinates $(1,0,0)$ in this basis.
A rotation of angle $\theta$ around an axis represented by the unit vector $n=(n_x,n_y,n_z)$ corresponds to the two following unitary quaternions~\cite{quaternions}: $$q= \pm\left(\:\cos (\theta / 2)+\sin (\theta / 2)\left(n_x \mathbf{i}+n_y \mathbf{j}+n_z \mathbf{k}\right)\:\right)$$Applying this formula in our case, we find that the corresponding quaternions are: $$\pm q= \pm\left(\cos \left(\frac{\pi}{2}\right)+\sin \frac{\pi}{2} \mathbf{i}\right)= \pm \mathbf{i}$$

In the same way, it is straightforward to show that a $\pm \mathbf{k}$ charge corresponds to an adjacent Dirac point.
Then, a loop encircling a pair of each type of Dirac points can be written as a composition of one loop encircling a principal point, and another encircling an adjacent point. By the product rule for quaternions, the corresponding charge is $\pm \mathbf{ik} = \pm \mathbf{j}$.

If a loop encircles no Dirac point, or two Dirac points which can be annihilated inside the loop, then it is contractible, therefore trivial, and its quaternionic charge is $1$. So we already can say that if the charge is different from 1, it means the Dirac points cannot be annihilated inside the loop.

In general for a loop encircling two principal Dirac points, again with the product rule, we get that the quaternionic charge is $\pm 1$. 

Finally, the remaining case to prove is that, when the two principal Dirac points cannot be annihilated, then the charge is $-1$. This can be seen by observing that when the charge is 1, then the loop is contractible, therefore the Dirac points can be annihilated.

\end{appendices}

\bibliography{NewEulerKagome}

\end{document}